\documentclass[12pt]{amsart}
 \usepackage[foot]{amsaddr}
\usepackage[top=0.9in, bottom=1in, left=1in, right=1in]{geometry}
\usepackage{amsmath, amssymb, amsthm}
\usepackage{verbatim}
\usepackage{graphicx}
\usepackage{enumerate}
\usepackage{mathrsfs}
\usepackage{tikz}
\usepackage{tikz-cd}
\usepackage{centernot}
\usepackage{hyperref}

\newtheorem{thm}{Theorem}[section]
\newtheorem{prop}[thm]{Proposition}
\newtheorem{lemma}[thm]{Lemma}
\newtheorem{cor}[thm]{Corollary}

     \newcommand{\sO}{\mathcal O}

\newcommand{\fr}{\mathfrak r}
\newcommand{\fn}{\mathfrak n} 
\newcommand{\afmsum}{\; \hat + \;}
\newcommand{\afmdiff}{\; \hat - \;}
\newcommand{\afmprod}{\; \hat \cdot \;}
\newcommand{\afm}{\mathscr M _{\textrm{aff}}}

\newcommand{\afr}{\mathscr R _{\textrm{aff}}}

\def\XXint#1#2#3{{\setbox0=\hbox{$#1{#2#3}{\int}$ }
\vcenter{\hbox{$#2#3$ }}\kern-.6\wd0}}

\newcommand{\R}{\mathbb{R}}
\newcommand{\Z}{\mathbb{Z}}
\newcommand{\C}{\mathbb{C}}

\newcommand{\N}{\mathbb{N}}

\newcommand{\iu}{{i\mkern1mu}}

\theoremstyle{definition}
\newtheorem{definition}[thm]{Definition}
\newtheorem{example}[thm]{Example}

\theoremstyle{remark}

\newtheorem{remark}[thm]{Remark}

\theoremstyle{ser}
\newtheorem{qn}{Question}

\numberwithin{equation}{section}

\newtheoremstyle{ser}
{8pt}
{8pt}
{\it}
{}
{\sf}
{:}
{6mm}
{}

\newtheoremstyle{serr}
{8pt}
{8pt}
{\normalfont}
{}
{\sf}
{.}
{6mm}
{}

\theoremstyle{ser}
\newtheorem{claim}{Claim}

\theoremstyle{serr}
\newtheorem{claimpff}{Proof of Claim}

\newcommand\Wtilde[1]{\stackrel{\sim}{\smash{\mathscr{#1}}\hspace{0.12in}\rule[0in]{0pt}{1.15ex}}\hspace{-0.12in}}
\tikzset{node distance=2cm, auto}
\makeatletter
\newcommand*\bigcdot{\mathpalette\bigcdot@{.5}}
\newcommand*\bigcdot@[2]{\mathbin{\vcenter{\hbox{\scalebox{#2}{$\m@th#1\bullet$}}}}}
\makeatother

\begin{document}
\title{Matrix algebras over algebras of unbounded operators}
\dedicatory{To the fond memory of my teacher, Richard V. Kadison (1925--2018)}
\author{Soumyashant Nayak}
\address{Smilow Center for Translational Research\\
 University of Pennsylvania\\
  Philadelphia\\
   PA 19104\\
   ORCiD: 0000-0002-6643-6574}
\email{nsoum@pennmedicine.upenn.edu}
\urladdr{https://nsoum.github.io/} 


\begin{abstract}
Let $\mathscr{M}$ be a $II_1$ factor acting on the Hilbert space $\mathscr{H}$, and $\afm$ be the Murray-von Neumann algebra of closed densely-defined operators affiliated with $\mathscr{M}$. Let $\tau$ denote the unique faithful normal tracial state on $\mathscr{M}$. By virtue of Nelson's theory of non-commutative integration, $\afm$ may be identified with the completion of $\mathscr{M}$ in the measure topology. In this article, we show that $M_n(\afm) \cong M_n(\mathscr{M})_{\textrm{aff}}$ as unital ordered complex topological $*$-algebras with the isomorphism extending the identity mapping of $M_n(\mathscr{M}) \to M_n(\mathscr{M})$. Consequently, the algebraic machinery of rank identities and determinant identities are applicable in this setting. As a step further in the Heisenberg-von Neumann puzzle discussed by Kadison-Liu (SIGMA, 10 (2014), Paper 009), it follows that if there exist operators $P, Q$ in $\afm$ satisfying the commutation relation $Q \afmprod P \afmdiff P \afmprod Q  = \iu I$, then at least one of them does not belong to $L^p(\mathscr{M}, \tau)$ for any $0 < p \le \infty$. Furthermore, the respective point spectrums of $P$ and $Q$ must be empty. Hence the puzzle may be recasted in the following equivalent manner - Are there invertible operators $P, A$ in $\mathscr{M}_{\textrm{aff}}$ such that $P^{-1} \afmprod A \afmprod P = I \afmsum A$? This suggests that any strategy towards its resolution must involve the study of conjugacy invariants of operators in $\mathscr{M}_{\textrm{aff}}$ in an essential way.

\vskip 0.01in \noindent
{\bf MSC2010 subject classification:} 47C15, 46L10, 46N50, 47L90
\end{abstract}

\maketitle

\section{Introduction}
\label{sec:intro}
Let $\hbar = \frac{h}{2\pi}$ denote the reduced Planck's constant, and $Q, P$ denote observables corresponding to a particle's position and its corresponding conjugate momentum, respectively. The Heisenberg commutation relation, $$PQ - QP = -\iu \hbar I,$$ is one of the most fundamental relations of quantum mechanics. It reveals the importance of non-commutativity in any foundational mathematical theory of quantum mechanics. Naturally it is of great interest to study non-commutative mathematical structures where the commutation relation may be represented. (For our study, we may normalize $\hbar$ to $1$.) For $n \in \mathbb{N}$, the algebra of $n \times n$ complex matrices $M_n(\mathbb{C})$ does not suffice because $\textrm{tr} (QP -PQ) = 0$ whereas $\textrm{tr}(\iu I)  = \iu $, where tr denotes the trace functional. It is known that the Heisenberg relation cannot be represented in any complex unital Banach algebra $\mathfrak{B}$ as $\sigma(AB) \cup \{0\} = \sigma(BA) \cup \{0\}$ (cf.\ \cite[Remark 3.2.9]{kadison-ringrose1}) for $A, B \in \mathfrak{B}$ where $\sigma(\cdot)$ denotes the spectrum of an element. This rules out the possibility of representing the Heisenberg relation using bounded operators on a complex Hilbert space. In the Dirac-von Neumann formulation of quantum mechanics (cf.\ \cite{dirac-quant-mech}, \cite{vN-math-found-qm}), quantum mechanical observables are defined as (possibly unbounded) self-adjoint operators on a complex Hilbert space. In this article, the adjective `unbounded' is used in the sense of being `not necessarily bounded' rather than `not bounded'. The classic representation of the Heisenberg relation (cf.\ \cite{vN-heisenberg-classic}, \cite[\S 5.3]{kadison-liu}) involves modeling $Q$ and $P$ as follows:
\begin{itemize}
\item[(i)] the position observable $Q$ is modeled by the unbounded self-adjoint operator $M$ defined as $(Mf)(x) = xf(x)$ $(x \in \R)$, with the domain being the set of functions $f$ in $L^2(\R)$ such that $Mf$ belongs to $L^2(\R)$;
\item[(ii)] the momentum observable $P$ is modeled by the unbounded self-adjoint operator $D = \iu \frac{\mathrm{d}}{\mathrm{d}x}$, with the domain being the linear subspace of $L^2(\R)$ corresponding to absolutely continuous functions on $\R$ whose almost-everywhere derivatives belong to $L^2(\R)$ (see \cite[Theorem 5.11]{kadison-liu}).
\end{itemize}
With this description, $QP - PQ$ is a pre-closed operator with closure $\iu I.$ But performing algebraic computations in this framework is an arduous task as one has to indulge in `domain-tracking'. We remind the reader that two unbounded operators that agree on a dense subspace can be very different in terms of the physics they describe (see \cite[pg.\ 254, Example 5]{simon-reed}).

Let $\mathscr{R}$ be a finite von Neumann algebra acting on the complex Hilbert space $\mathscr{H}$, and let $\mathscr{R}_{\textrm{aff}}$ denote the set of closed densely-defined operators affiliated with $\mathscr{R}$.  By virtue of \cite[Theorem 6.13]{kadison-liu}, $\mathscr{R}_{\textrm{aff}}$ has the structure of a unital $*$-algebra, and is called the Murray-von Neumann algebra associated with $\mathscr{R}$ (see \cite[Definition 6.14]{kadison-liu}). In the case of a finite factor, the result was first observed by Murray and von Neumann in \cite[Theorem XV]{vN-rings}. When $\mathscr{R}$ is countably decomposable and thus possesses a faithful normal tracial state, the fact that $\mathscr{R}_{\textrm{aff}}$ is a unital $*$-algebra also follows from Theorem 4 in \cite{nelson}. In \cite[\S 7]{kadison-liu}, Kadison and Liu discuss the Heisenberg-von Neumann puzzle which may be stated as follows: Is there a representation of the Heisenberg commutation relation in $\mathscr{R}_{\textrm{aff}}$?  In \cite[Theorem 7.4]{kadison-liu}, they showed that one cannot use self-adjoint operators in $\mathscr{R}_{\textrm{aff}}$ to represent the Heisenberg relation. The key step in the proof involves the use of the center-valued  trace on $\mathscr{R}$ after `wrestling' the unbounded operators down to the bounded level (see \cite[Lemma 7.3, Theorem 7.4]{kadison-liu}) using the spectral decomposition for unbounded self-adjoint operators. So one may speculate that although there is no obvious trace on $\mathscr{R}_{\textrm{aff}}$, perhaps the center-valued trace on $\mathscr{R}$ provides a moral obstruction to the Heisenberg commutation relation. But in \cite{kadison-liu-thom}, Kadison, Liu, and Thom have shown that the identity operator is the sum of two commutators in the Murray-von Neumann algebra associated with a type $II_1$ von Neumann algebra. Thus the moral argument clearly fails, leaving the question of representing the Heisenberg relation in $\mathscr{R}_{\textrm{aff}}$ with non-self-adjoint operators wide open.

We have reviewed a few arguments above that identify obstructions to representing the Heisenberg relation in the algebra under consideration. A common feature of these arguments is that they involve comparison of the spectral content of $QP$ and $PQ$ in one form or another. For unbounded operators, we may capture the spirit of these arguments by studying notions of rank and determinant which are available in certain algebras of unbounded operators. In this article, our main goal is to faciliate such a study by rigorously defining matrix algebras over some important classes of algebras of unbounded operators. In the literature, several proofs rely on matrix algebraic arguments (for example, see \cite[Lemma 2.21-2.23]{haagerup-schultz}, \cite[Proposition 2.4]{haagerup-schultz}) whose justification is not provided. Our investigation reveals that the justification of these arguments is not a trivial matter and requires an understanding of an appropriate topology on Murray-von Neumann algebras, namely the measure topology. The main goals of this article are the following:
\begin{itemize}
\item[(i)] To develop foundations for matrix theoretic arguments involving unbounded operators,
\item[(ii)] To extend the study of rank and determinant identities to the context of unbounded operators,
\item[(iii)] To apply the techniques developed to further study of the Heisenberg-von Neumann puzzle.
\end{itemize}

The set of $n \times n$ matrices $M_n(\mathfrak{A})$ over a $*$-algebra $\mathfrak{A}$ naturally has the structure of a $*$-algebra. Using the addition and multiplication in $\mathfrak{A}$, we may define the algebraic structure of $M_n(\mathfrak{A})$ through the usual addition and multiplication of matrices. For a matrix $\mathbf{T} \in M_n(\mathfrak{A})$ (with $T_{ij} \in \mathfrak{A}$ as its $(i, j)$th entry), the adjoint $\mathbf{T}^{\dagger}$ is the matrix whose $(i, j)$th entry is $T_{ji}^*$. In the case when $\mathfrak{A}$ is a $C^*$-algebra acting on the Hilbert space $\mathscr{H}$, one may represent $M_n(\mathfrak{A})$ faithfully on the Hilbert space $\oplus_{i=1}^n \mathscr{H}$, through the usual matrix action on column vectors. It is a basic (though not entirely trivial) result that the norm on $M_n(\mathfrak{A})$ inherited from $\mathcal{B}(\oplus_{i=1}^n \mathscr{H})$  makes it a $C^*$-algebra and this norm is independent of the representation of $\mathfrak{A}$ on $\mathscr{H}$ (see \cite[Proposition 11.1.2]{kadison-ringrose2}). Similarly for a von Neumann algebra $\mathscr{R}$ acting on the Hilbert space $\mathscr{H}$, $M_n(\mathscr{R})$ is a von Neumann algebra acting on $\oplus_{i=1}^n \mathscr{H}$.

The classical operator algebras ($C^*$-algebra, von Neumann algebras, respectively) are usually defined at the outset as $*$-closed subalgebras of the $*$-algebra of bounded operators on a Hilbert space (which are norm-closed, weak-operator closed, respectively). At the same time, there are many advantages to studying intrinsic characterizations of these algebras that are independent of the representation. Firstly, such a description helps in identifying the right notion of morphism in the class of operator algebras under consideration. For instance, an abstract $C^*$-algebra is defined as an involutive complex Banach algebra with the norm satisfying the $C^*$-axioms. The Gelfand-Neumark theorem shows that every abstract $C^*$-algebra has a faithful representation as a concrete $C^*$-algebra. The morphisms in the category of $C^*$-algebras are given by $*$-homomorphisms, which are automatically norm-continuous. Similarly, Sakai's theorem (cf.\ \cite{sakai}) shows that a von Neumann algebra may be characterized as a $C^*$-algebra which is the dual of a Banach space.\footnote{The ultraweak topology corresponds to the weak-$*$ topology induced by the pre-dual.} The morphisms in the category of von Neumann algebras are given by normal $*$-homomorphisms, which are automatically ultraweak continuous. Secondly, although several constructions (such as direct sums, tensor products, etc.) involving operator algebras are best understood using concrete representations, it is crucial to ensure that the constructed object is independent of the representations used for the building blocks. One such construction of fundamental importance is the process of forming matrices over operator algebras which we discussed in the previous paragraph.

Let $\mathscr{M}$ be a $II_1$ factor acting on the Hilbert space $\mathscr{H}$ with the unique faithful normal tracial state on $\mathscr{M}$  denoted by $\tau$. In this article, we are primarily interested in studying (full) matrix algebras over $\afm$. In \cite[\S 2]{nelson}, Nelson defined the {\it measure topology} on von Neumann algebras with a faithful normal semi-finite trace with the goal of studying a non-commutative version of the notion of {\it convergence in measure}. By \cite[Theorem 4]{nelson}, $\afm$ has an intrinsic description as the completion of $\mathscr{M}$ in the measure topology. In Proposition \ref{prop:mtilde_pos}, we show that this description also accommodates a natural order structure, with the cone given by the closure of the positive cone of $\mathscr{M}^{\mathrm{sa}}$ in the measure topology. Moreover, this intrinsic order structure is compatible with the usual order structure on the space of self-adjoint operators in $\afm$ given by the cone of positive affiliated operators (see Remark \ref{rmrk:maff_pos}). Although this exercise is of interest on its own, our main goal is to utilize such a description to answer the following question: Is $M_n(\afm) \cong M_n(\mathscr{M})_{\textrm{aff}}$ in a suitable sense?\footnote{Note that $M_n(\mathscr{M})$ is a $II_1$ factor.} In Theorem \ref{thm:main_iso}, we answer this question affirmatively by showing that  $M_n(\afm)$ and $M_n(\mathscr{M})_{\textrm{aff}}$ are isomorphic as unital ordered complex topological $*$-algebras, with the isomorphism extending the identity mapping of $M_n(\mathscr{M}) \to M_n(\mathscr{M})$.\footnote{By an ordered complex $*$-algebra, we mean a complex $*$-algebra whose Hermitian elements form an ordered real vector space.}

In Theorem \ref{thm:top_comp}, we collect several results concerning the measure topology, juxtaposing it  with the norm topology and the ultraweak topology. For instance, we note that the measure topology on $\mathscr{M}$ is finer than the norm topology but, in general, it may be neither coarser nor finer than the ultraweak topology. In Theorem \ref{thm:top_comp}, (iii),  we note that the trace $\tau$ is not continuous in the measure topology on $\mathscr{M}$ and thus, does not (measure) continuously extend to a trace on $\afm$. This is not entirely surprising since if $\tau$ had such an extension, it would  contradict the previously discussed result by Kadison, Liu, and Thom in \cite{kadison-liu-thom}. In light of the above observation, it is worth mentioning that the restriction of the trace $\tau$ to the unit ball of $\mathscr{M}$ is continuous in the measure topology.

Let $\Delta, \Delta_n$ denote the Fuglede-Kadison determinant of $\mathscr{M}, M_n(\mathscr{M})$, respectively. Let $\mathscr{M}_{\Delta}$ be the set of operators $T$ in $\afm$ satisfying $\tau(\log ^{+} |T|) < \infty$, where $\log ^{+} := \max \{ 0, \log \}$ on $\R _{+}$.\ In other words, $\mathscr{M}_{\Delta}$ consists of those operators in $\afm$ whose Fuglede-Kadison determinant is bounded. Haagerup and Schultz showed that $\mathscr{M}_{\Delta}$ is a $*$-subalgebra of $\afm$ (see \cite[Proposition 2.5-2.6]{haagerup-schultz}), and studied the notion of Brown measure for operators in $\mathscr{M}_{\Delta}$.  The space $\mathscr{M}_{\Delta}$ is quite rich, containing the non-commutative $L^p$-spaces, $L^p(\mathscr{M}; \tau)$, for $0 < p \le \infty$. In Proposition \ref{prop:doug_lem}, we prove a version of the Douglas factorization lemma (cf.\ \cite{douglas-factor}) in the context of Murray-von Neumann algebras. With the help of this generalized Douglas lemma, in Theorem \ref{thm:main_iso_delta}, we show that $M_n(\mathscr{M}_{\Delta}) = M_n(\mathscr{M})_{\Delta _n}$, viewing both $*$-algebras as $*$-subalgebras of $M_n(\mathscr{M})_{\textrm{aff}}$. Together with Theorem \ref{thm:main_iso}, this result enables the use of matrix algebraic techniques such as rank identities and determinant identities in $\afm$ and $\mathscr{M}_{\Delta}$, respectively. In \S \ref{sec:heisenberg}, we apply the results to provide some necessary conditions for pairs of operators in $\afm$ satisfying the Heisenberg relation. In Theorem \ref{thm:bmeas_ab_ba},(ii), with the aid of Sylvester's determinant identity, we show that for operators $A, B$ in $\mathscr{M}_{\Delta}$, the Brown measures of $A \afmprod B$ and $B \afmprod A$ are identical. Although a stronger result may be inferred from Theorem A. 9 in \cite{brown-lidskii}, our novel approach serves as a proof-of-principle for the application of matrix identities to unbounded operators.  As a corollary (Corollary \ref{cor:heisenberg_lp}), we note that if $Q \afmprod P \afmdiff P \afmprod Q = \iu  I$ for operators $P, Q$ in $\afm$, then at least one of $P, Q$ does not belong to $\mathscr{M}_{\Delta}$ and {\it a fortiori}, does not belong to $L^p(\mathscr{M}; \tau)$ for any $0 < p \le \infty$. 

For an operator $T \in \afm$, we may define its rank by $\fr(T) := \tau(\mathcal{R}(T)) \in [0, 1]$, where $\mathcal{R}(T)$ denotes the range projection of $T$, which is a projection in $\mathscr{M}$. One of the key properties of the rank functional is that for an invertible operator $S$ in $\afm$, $\fr(S \afmprod T) = \fr(T \afmprod S) = \fr(T)$ for all $T \in \afm$. In Corollary \ref{cor:heisenberg_drange}, using a matrix identity we show that, if $Q \afmprod P \afmdiff P \afmprod Q = \iu I$ for operators $P, Q$ in $\afm$, then for all $\lambda \in \mathbb{C}$ the operators $P - \lambda I$ and $Q - \lambda I$ are invertible in $\afm$, that is, the respective point spectrums of $P$ and $Q$ must be empty. Hence the Heisenberg-von Neumann puzzle may be recasted in the following equivalent manner: Are there invertible operators $P, A$ in $\mathscr{M}_{\textrm{aff}}$ such that $$P^{-1} \afmprod A \afmprod P = I \afmsum A?$$
This suggests that any strategy towards the resolution of the Heisenberg-von Neumann puzzle must involve conjugacy invariants of operators in $\mathscr{M}_{\textrm{aff}}$ in an essential way.

\subsection{Notation and Terminology}
Throughout this article, $\mathscr{H}$ denotes a Hilbert space over the complex numbers (usually infinite-dimensional, though not necessarily separable). For a positive integer $n$, the Hilbert space $\oplus_{i=1}^n \mathscr{H}$ is denoted by $\mathscr{H}^{(n)}$. A bounded operator $A : \mathscr{H} \to \mathscr{H}$ is said to be a {\it contraction} if $\|A\| \le 1$. We use $\mathscr{R}$ to denote a von Neumann algebra, and $\mathscr{M}$ to denote a $II_1$ factor. The unique faithful normal tracial state on $\mathscr{M}$ is denoted by $\tau$. The normalized dimension function for projections in $\mathscr{M}$ is denoted by $\mathrm{dim}_c(\cdot).$ A complex $*$-algebra $\mathfrak{A}$ is said to be {\it ordered} if the Hermitian elements in $\mathfrak{A}$ form an ordered real vector space. For an ordered complex $*$-algebra $\mathfrak{A}$ (such as von Neumann algebras, Murray-von Neumann algebras, etc.), we denote the set of self-adjoint elements in $\mathfrak{A}$ by $\mathfrak{A}^{\mathrm{sa}}$, and the positive cone of $\mathfrak{A}^{\mathrm{sa}}$ by $\mathfrak{A}^{+}$. For a matrix $\mathbf{T}$ in $M_n(\mathfrak{A})$, we denote the matrix adjoint of $\mathbf{T}$ by $\mathbf{T}^{\dagger}$. The identity operator in $\mathfrak{A}$ is denoted by $I$ and the identity matrix of $M_n(\mathfrak{A})$ is denoted by $\textbf{I}_n$. We denote a net of operators by $\{ T_{\alpha} \}$ suppressing the indexing set of $\alpha$ (denoted by $\Lambda$) when it is clear from the context. The general references used are \cite{kadison-ringrose1}, \cite{kadison-ringrose2}.

\subsection{Acknowledgments}
This article is dedicated to my former advisor, Richard V. Kadison, who passed way in August 2018. In the last phase of his (mathematical) life, the topic of Murray-von Neumann algebras was dear to his heart. His vision of the field and encouraging words in regards to a preliminary version of the ideas presented herein serve as an inspiration for this work. It is a pleasure to thank Amudhan Krishnaswamy-Usha for helpful discussions on the Brown measure at ECOAS 2018, and Konrad Schrempf for ongoing discussions on free associative algebras over fields. I am also grateful to Zhe Liu for valuable feedback regarding an early draft of the article.

\section{Murray-von Neumann algebras}

\subsection{Unbounded Operators on a Hilbert space}

In this subsection, we provide a brief overview of the basic concepts and results in the theory of unbounded operators that are directly relevant to our discussion. A concise account can be found in \cite[\S 4]{kadison-liu}. For a more thorough account, the interested reader may refer to \S 2.7, \S 5.6 in \cite{kadison-ringrose1}, or Chapter VIII in \cite{simon-reed}.

Let $\mathscr{H}$ be a Hilbert space and let $T$ be a linear mapping from a linear submanifold $\mathscr{D}(T)$ of $\mathscr{H}$ (not necessarily closed), called the {\it domain} of $T$,  into $\mathscr{H}$. The linear submanifold $\mathscr{G}(T) := \{ (x, Tx) : x \in \mathscr{D}(T) \}$ of $\mathscr{H} \oplus \mathscr{H}$ is said to be the {\it graph} of $T$. We say that $T$ is {\it closed} if $\mathscr{G}(T)$ is a closed linear submanifold of $\mathscr{H} \oplus \mathscr{H}$. From the Closed Graph Theorem, if $T$ is closed and $ \mathscr{D}(T) = \mathscr{H}$, then $T$ is bounded. The property of being closed often serves as a replacement for continuity in the study of unbounded operators.  We are usually interested in operators $T$ that are {\it densely defined}, that is, $\mathscr{D}(T)$ is dense in $\mathscr{H}$. An operator $T_0$ is said to be an {\it extension} of $T$ (denoted $T \subseteq T_0$), if $\mathscr{D}(T) \subseteq \mathscr{D}(T_0)$ and $Tx = T_0 x$ for $x \in \mathscr{D}(T)$. If the closure of $\mathscr{G}(T)$ in $\mathscr{H} \oplus \mathscr{H}$ is the graph of an operator $\overline{T}$, then $T$ is said to be {\it pre-closed} or {\it closable} with closure $\overline{T}$. For a closed operator $T$, a linear submanifold $\mathscr{D}_0$ of $\mathscr{D}(T)$ is said to be a {\it core} for $T$ if $\mathscr{G}(T|_{\mathscr{D}_0})^{-} = \mathscr{G}(T)$; the operator $T$ maps a core onto a dense linear submanifold of its range.

The sum of unbounded operators $S, T$ on $\mathscr{H}$ is defined as the operator $S+T$ with $\mathscr{D}(S+T) = \mathscr{D}(S) \cap \mathscr{D}(T)$, and $(S+T)x = Sx + Tx$  for $x \in \mathscr{D}(S+T)$. The product of $S, T$ is defined as the operator $ST$ with $\mathscr{D}(ST) = \{ x \in \mathscr{H} : x \in \mathscr{D}(T) \textrm{ and } Tx \in \mathscr{D}(S) \}$, and $(ST)x = S(Tx)$ for $x \in \mathscr{D}(ST).$ The adjoint of a densely-defined operator $T$ is defined as the operator $T^*$ with $\mathscr{D}(T^*) = \{ x \in \mathscr{H} : y \in \mathscr{D}(T) \mapsto \langle Ty, x \rangle \textrm{ is a bounded functional} \}$ and for $x \in \mathscr{D}(T^*)$, $T^*x = z$ where $\langle Ty, x \rangle = \langle y, z \rangle$ for all $y \in \mathscr{D}(T)$ (such a $z$ exists by the Riesz representation theorem.) If $T = T^*$, we say that $T$ is self-adjoint. A self-adjoint operator is automatically closed.

\begin{definition}
A closed densely-defined operator $T$ is said to be {\it affiliated} with a von Neumann algebra $\mathscr{R}$ if $U^*TU = T$ for each unitary operator $U$ in $\mathscr{R}'$, the commutant of $\mathscr{R}$. We write this as $T \eta \mathscr{R}$. Note that the equality $U^*TU = T$ carries the information that $U$ transforms the domain $\mathscr{D}(T)$ onto itself. We denote the set of closed densely-defined operators affiliated with $\mathscr{R}$ by $\mathscr{R}_{\textrm{aff}}$.
\end{definition}

\begin{definition}
A self-adjoint operator $H$ is said to be {\it positive} if $0 \le \langle Hx, x \rangle$ for all $x \in \mathscr{D}(H)$. 
\end{definition}

\begin{definition}
Let $T$ be a closed densely-defined operator affiliated with a von Neumann algebra $\mathscr{R}$.  The projection onto the closure of the range of $T$ is said to be the range projection of $T$, and denoted by $\mathcal{R}(T)$. The projection onto the null space of $T$ is denoted by $\mathcal{N}(T)$.

The range projection of $T$ is the smallest projection in $\mathscr{R}$ amongst all projections $E$ in $\mathscr{R}$ satisfying $ET = T$.
\end{definition}

\begin{prop}[see {\cite[Proposition 4.7, 6.5]{kadison-liu}}]
\label{prop:range_proj}
\textsl{
Let $\mathscr{R}$ be a von Neumann algebra acting on a Hilbert space $\mathscr{H}$, and $\sim$ denote the Murray-von Neumann equivalence relation (relative to $\mathscr{R}$) on the set of projections in $\mathscr{R}$. If $T$ is a closed densely-defined operator affiliated with $\mathscr{R}$, then:
\begin{itemize}
\item[(i)] $\mathcal{R}(T) = I - \mathcal{N}(T^*)$ and $\mathcal{N}(T) = I - \mathcal{R}(T^*)$;
\item[(ii)] $\mathcal{R}(T) = \mathcal{R}(TT^*)$ and $\mathcal{N}(T) = \mathcal{N}(T^*T)$;
\item[(iii)] $\mathcal{R}(T)$ and $\mathcal{N}(T)$ are in $\mathscr{R}$;
\item[(iv)] $\mathcal{R}(T) \sim \mathcal{R}(T^*)$ relative to $\mathscr{R}$.
\end{itemize}
}
\end{prop}

\subsection{Murray-von Neumann algebras}

In this subsection, $\mathscr{R}$ denotes a finite von Neumann algebra acting on the Hilbert space $\mathscr{H}$. In \cite[Theorem XV]{vN-rings}, Murray and von Neumann observed that when $\mathscr{R}$ is a finite factor, $\mathscr{R}_{\textrm{aff}}$ may be endowed with the structure of a $*$-algebra. When $\mathscr{R}$ is countably decomposable (and thus, possesses a faithful normal tracial state), it follows from the work of Nelson (cf.\ \cite[Theorem 4]{nelson}) that a similar conclusion holds.

\begin{prop}[see {\cite[Proposition 6.8]{kadison-liu}}]
\label{prop:fund_mva}
\textsl{
Let $\mathscr{R}$ be a finite von Neumann algebra acting on the Hilbert space $\mathscr{H}$. For operators $A, B$ in $\mathscr{R}_{\textrm{aff}}$, we have:
\begin{itemize}
\item[(i)] $A+B$ is densely defined, preclosed and has a unique closed extension $A \afmsum B$ in $\mathscr{R}_{\textrm{aff}}$;
\item[(ii)] $AB$ is densely defined, preclosed and has a unique closed extension $A \afmprod B$ in $\mathscr{R}_{\textrm{aff}}$.
\end{itemize}
}
\end{prop} 

 In \cite[Proposition 6.9 - 6.12]{kadison-liu}, Kadison and Liu showed that for a general finite von Neumann algebra $\mathscr{R}$, the set of affiliated operators $\mathscr{R}_{\textrm{aff}}$ is a $*$-algebra with $\afmsum$ as addition, $\afmprod$ as multiplication, and $T \mapsto T^*$ as the involution. This was accomplished by carefully studying and tracking the domains of the operators under consideration to prove the associative, distributive and involutive laws involving $\afmsum, \afmprod$, and $(\cdot)^*$.   
 
Although it may be tempting to replace the symbols $\afmsum, \afmprod$ with $+, \cdot$ once the algebraic structure of $\mathscr{R}_{\textrm{aff}}$ has been established, we refrain from doing so in this article as $+, \cdot$ already have pre-defined meanings for unbounded operators. For a bounded operator $B$ and a closed densely-defined operator $T$ in $\mathscr{R}_{\textrm{aff}}$, it is straightforward to see that $BT = B \afmprod T$ and $B + T = B \afmsum T$, but $TB$ is not necessarily equal to $T \afmprod B$.

\begin{definition}
For a finite von Neumann algebra $\mathscr{R}$, the $*$-algebra of affiliated operators $\mathscr{R}_{\textrm{aff}}$ is called the {\it Murray-von Neumann algebra} associated with $\mathscr{R}$. 
\end{definition}

The set of positive operators in $\mathscr{R}_{\textrm{aff}}$ is a cone and with this positive cone, $\mathscr{R}_{\textrm{aff}}$ may be viewed as an ordered $*$-algebra. In other words, for self-adjoint operators $A, B \in \mathscr{R}_{\textrm{aff}}$, we say that $A \le B$ if $B \afmdiff A$ is a positive operator.

\begin{definition}
\label{def:rank}
Let $\mathscr{R}$ be a finite von Neumann algebra with center $\mathscr{C}$ and let $\tau$ denote the center-valued trace. For a closed densely-defined operator $T$ in $\mathscr{R}_{\textrm{aff}}$, the $\mathscr{C}$-valued {\it rank} of $T$ is defined as $\fr (T) := \tau(\mathcal{R}(T)).$ In other words, the rank of $T$ is the $\mathscr{C}$-valued dimension of the range projection of $T$.
\end{definition}

\begin{prop}
\label{prop:rank_prop}
\textsl{
Let $\mathscr{R}$ be a finite von Neumann algebra acting on the Hilbert space $\mathscr{H}$ and $S, T$ be operators in $\mathscr{R}_{\textrm{aff}}$ such that $S$ is invertible in $\mathscr{R}_{\textrm{aff}}$. Then $\fr (S \afmprod T) = \fr (T \afmprod S) = \fr (T).$
} 
\end{prop}
\begin{proof}
Note that the operators $S, S^*$ are invertible in $\mathscr{R}_{\textrm{aff}}$. It is easy to see that $\mathcal{N}(T^*) \le \mathcal{N}(S^* \afmprod T^*) \le \mathcal{N}((S^*)^{-1} \afmprod (S^* \afmprod T^*)) = \mathcal{N}(T^*).$ Thus $\mathcal{N}(S^* \afmprod T^*) = \mathcal{N}(T^*)$. By Proposition \ref{prop:range_proj},(i), we have $\mathcal{R}(T \afmprod S) = \mathcal{R}(T)$ which implies that $ \fr (T \afmprod S) = \fr (T).$ Similarly $\mathcal{R}(T^* \afmprod S^*) = \mathcal{R}(T^*)$. From Proposition \ref{prop:range_proj},(iv), we conclude that $\mathcal{R}(S \afmprod T) \sim \mathcal{R} (T)$ which implies that $\fr (S \afmprod T) = \fr (T).$
\end{proof}

\begin{prop}[see {\cite[Lemma 8.20, Theorem 8.22]{luck}}]
\label{prop:luck}
\textsl{
Let $A$ be an operator in the ring $\afr$. Then the following are equivalent:
\begin{itemize}
    \item[(i)] $A$ is not a left zero-divisor;
    \item[(ii)] $A$ is not a zero-divisor;
    \item[(iii)] $A$ is invertible;
    \item[(iv)] $A$ has dense range;
    \item[(v)] $A$ has trivial nullspace;
    \item[(vi)] $\fr (A) = I$.
\end{itemize}
}
\end{prop}

\begin{lemma}
\label{lem:dense_dom}
\textsl{
Let $\mathscr{R}$ be a finite von Neumann algebra acting on the Hilbert space $\mathscr{H}$. For operators $A_1, \cdots, A_n$ in $\mathscr{R}_{\textrm{aff}}$, the linear manifold $\bigcap_{i=1}^n \mathscr{D}(A_i)$ is a core for each of $A_1, \cdots, A_n$.
}
\end{lemma}
\begin{proof}
Since the key ingredients for a proof are already present in the proofs of \cite[Proposition 6.8 - 6.9]{kadison-liu}, we only provide an outline for an argument. For a positive integer $m$, let $T_1, \cdots, T_m$ be operators in $\mathscr{R}_{\textrm{aff}}.$ Note that the operators $T_1 + \cdots + T_m$ and $T_1 ^* + \cdots + T_m ^*$ are both densely defined (cf.\ \cite[Proposition 6.8,(i)]{kadison-liu}). As $T_1 ^* + \cdots + T_m ^* \subseteq (T_1 + \cdots + T_m)^*$, $T_1 + \cdots + T_m$ is preclosed and thus $\bigcap_{i=1}^{m}\mathscr{D}(T_i)$ is a core for $T_1 \afmsum \cdots \afmsum T_m$. Take $T_1 = A_1, \cdots, T_n = A_n, T_{n+1} = -A_n, T_{n+2} = -A_{n-1}, \cdots, T_{2n-1} = -A_2.$ Thus $\bigcap_{i=1}^n \mathscr{D}(A_i) = \bigcap_{i=1}^{2n-1} \mathscr{D}(T_i)$ is a core for $(A_1 \afmsum (A_2 \afmdiff A_2) \afmsum \cdots \afmsum (A_n \afmdiff A_n)) = A_1$. Similarly $\bigcap_{i=1}^n \mathscr{D}(A_i)$ is a core for each of $A_2, \cdots, A_n.$
\end{proof}

\begin{prop}[Douglas factorization lemma]
\label{prop:doug_lem}
\textsl{
Let $\mathscr{R}$ be a finite von Neumann algebra acting on the Hilbert space $\mathscr{H}$. For $A, B \in \mathscr{R}_{\textrm{aff}}$, the following are equivalent:
\begin{itemize}
\item[(i)] $A^* \afmprod A \le  B^* \afmprod B$,
\item[(ii)] $A = CB$ for $C$ in $\mathscr{R}$ with $\| C \| \le 1$.
\end{itemize}
}
\end{prop}
\begin{proof}
(i) $\Rightarrow$ (ii).\\
Let $\mathscr{V} := \mathscr{D}(A) \cap \mathscr{D}(B) \cap  \mathscr{D}(A^* \afmprod A) \cap \mathscr{D}(B^* \afmprod B) \subseteq \mathscr{H}$ . By Lemma \ref{lem:dense_dom}, note that $\mathscr{V}$ is a core for each of $A, B, A^* \afmprod A, B^* \afmprod B$. Let us denote the range of $B$ by $\mathrm{ran}(B) \subseteq \mathscr{H}$.  For a vector $x \in \mathscr{V}$, we define $C(Bx) = Ax$, and for a vector $y$ in $\mathrm{ran}(B)^{\perp}$, we define $Cy = 0$. Since $\mathscr{V}$ is a core for $B$, we observe that $B$ maps $\mathscr{V}$ onto a dense subset of $\mathrm{ran}(B)$. For $x \in \mathscr{V}$ and $y$ in $\mathrm{ran}(B)^{\perp}$, as $\|C(Bx + y) \|^2 = \|Ax \|^2 = \langle A^* \afmprod A x, x \rangle \le \langle B^* \afmprod Bx, x \rangle = \|Bx\|^2 \le \|Bx\|^2 + \|y\|^2 = \|Bx + y\|^2$, we conclude that $C$ is a bounded operator such that $A=CB$ and $\| C \| \le 1$.

Let $U$ be a unitary operator in the commutant $\mathscr{R}'$ of $\mathscr{R}$. Then $UA = AU, U^*A = AU^*, UB = BU, U^*B=BU^*$ and the linear subspace $\mathrm{ran}(B)$ is invariant under $U$ and so is the closed subspace $\mathrm{ran}(B)^{\perp}$. For vectors $x_1$ in $\mathscr{V}$ and $x_2$ in $\mathrm{ran}(B)^{\perp}$, we have $CU(Bx_1 + x_2) = CUBx_1 + C(Ux_2) = CB(Ux_1) + 0 = A(Ux_1) = U(Ax_1) = UCBx_1 = UC(Bx_1 + x_2)$. Thus $UC$ and $CU$ coincide on the dense subspace of $\mathscr{H}$ given by $B\mathscr{V} \oplus \mathrm{ran}(B)^{\perp}$. Since $UC$ and $CU$ are bounded operators, we note that $UC = CU$ for any unitary operator $U$ in $\mathscr{R}'$. As every element in a von Neumann algebra can be written as a linear combination of four unitary elements, we conclude that $C$ commutes with every element in $\mathscr{R}'$. By the double commutant theorem, $C$ is in $(\mathscr{R}')' = \mathscr{R}$.
\vskip 0.1in
\noindent (ii) $\Rightarrow$ (i).\\
If $A=CB$ for a contraction $C \in \mathscr{R}$, then $A^* \afmprod A = B^* \afmprod (C^*C) \afmprod B$. As $I - C^*C$ is a positive operator, $ B^* \afmprod (I-C^*C) \afmprod B = B^* \afmprod B - A^* \afmprod A$ is a positive operator in $\mathscr{R}_{\textrm{aff}}$.
\end{proof}

\subsection{Operators in $\mathscr{M}_{\Delta}$}

Let $\mathscr{M}$ be a $II_1$ factor acting on the Hilbert space. Let $\tau$ be the unique faithful normal tracial state on $\mathscr{M}$. An operator $T \in \afm$ has a unitary polar decomposition (see \cite[Theorem 6.1.11]{kadison-ringrose2}, \cite[Exercise 6.9.6]{kadison-ringrose2}), $$T = U|T| =U \left( \int_{0}^{\infty} \lambda \; dE_{|T|}(\lambda) \right),$$ where $U \in \mathscr{M}$ is unitary, and the spectral measure $E_{|T|}$ takes values in $\mathscr{M}$. We may define a Borel probability measure $\mu_{|T|}$ on $\mathbb{R}_{+}$ by $$\mu_{|T|}(S) = \tau(E_{|T|}(S)),$$ for a Borel set $S \subseteq \mathbb{R}_{+}$. 
\begin{definition}
Let $\log ^{+}$ be the function on $\R_{+}$ defined by $\max \{0 , \log \}$.
We define $\mathscr{M}_{\Delta}$ to be the set of operators $T \in \afm$ satisfying the condition $$\tau(\log ^{+} |T|) = \int_{0}^{\infty} \log ^{+} \lambda \; d\mu_{|T|}(\lambda) = \int_{1}^{\infty} \log \lambda \; d\mu_{|T|}(\lambda)< \infty.$$ 
Thus for $T \in \mathscr{M}_{\Delta}$, we have $$-\infty \le \int_{0}^{\infty} \log \lambda \; d\mu_{|T|}(\lambda) < \infty.$$
The {\it Fuglede-Kadison} determinant of $T$ is defined as $$\Delta(T) := \exp \Big( \int_{0}^{\infty} \log \lambda \; d\mu_{|T|}(\lambda)  \Big).$$
\end{definition}

\begin{lemma}[see {\cite[Remark 2.2]{haagerup-schultz}}]
\label{lem:mdelta_desc}
\textsl{
$L^p(\mathscr{M}, \tau) \subset \mathscr{M}_{\Delta}$ for all $p \in (0, \infty].$ In particular, $\mathscr{M} \subset \mathscr{M}_{\Delta}.$
}
\end{lemma}

\begin{lemma}[see {\cite[Proposition 2.5, 2.6]{haagerup-schultz}}]
\label{lem:mdelt_alg}
\textsl{
For operators $S, T \in \mathscr{M}_{\Delta}$,  we have:
\begin{itemize}
\item[(i)] $S \afmprod T \in \mathscr{M}_{\Delta}$,
\item[(ii)] $S \afmsum T \in \mathscr{M}_{\Delta}$,
\item[(iii)] $S^* \in \mathscr{M}_{\Delta}$.
\end{itemize} 
Thus $\mathscr{M}_{\Delta}$ is a $*$-subalgebra of $\afm$, containing $\mathscr{M}$.
}
\end{lemma}

\begin{definition}
Let $T$ be an operator in $\mathscr{M}_{\Delta}$ and let $f$ be the mapping $\lambda \in \mathbb{C} \to \log \Delta(\lambda I - T) \in \mathbb{C}.$ Then by \cite[Theorem 2.7]{haagerup-schultz}, $f$ is a subharmonic function on $\mathbb{C}$, and $$d\mu_T = \frac{1}{2 \pi} \nabla ^2 f \; d\lambda$$
(in the distributional sense) defines a Borel probability measure on $\mathbb{C}$ which is called the {\it Brown measure} of $T$.
\end{definition}

\begin{remark}
\label{rmrk:mdelta_pos}
Let $\mathscr{M}_{\Delta}^{\textrm{sa}}$ be the set of self-adjoint operators in $\mathscr{M}_{\Delta}$, and $\mathscr{M}_{\Delta}^{+}$ be the set of positive operators in $\mathscr{M}_{\Delta}$.
For a positive operator $H \in \mathscr{M}_{\Delta}$, since $\log ^{+} (\sqrt{H}) = \frac{1}{2} \log ^{+}(H)$, we have $\sqrt{H} \in \mathscr{M}_{\Delta}.$ Thus $\mathscr{M}_{\Delta}^{+} := (\mathscr{M}_{\Delta}) ^{\mathrm{sa}} \cap \afm ^{+} = \{ H^2 : H \in \mathscr{M}_{\Delta}^{\textrm{sa}} \}$, which gives an algebraic description of the positive cone of $(\mathscr{M}_{\Delta})^{\mathrm{sa}}$ inherited from $(\afm) ^{\mathrm{sa}}$. 
\end{remark}

\begin{prop}
\label{prop:ord_ideal_mdelta}
\textsl{
For operators $A \in \, (\afm) ^{\mathrm{sa}}$ and $B \in (\mathscr{M}_{\Delta}) ^{\mathrm{sa}}$, if $0 \le A \le B$, then $A \in (\mathscr{M}_{\Delta}) ^{\mathrm{sa}}$.
}
\end{prop}
\begin{proof}
By the Douglas lemma (Proposition \ref{prop:doug_lem}), there is a bounded operator $C \in \mathscr{M}$ such that $\sqrt{A} = C \sqrt{B}$. By Remark \ref{rmrk:mdelta_pos} and Lemma \ref{lem:mdelt_alg}, we have $\sqrt{B}\in \mathscr{M}_{\Delta} \Longrightarrow \sqrt{A} \in \mathscr{M}_{\Delta} \Longrightarrow A \in \mathscr{M}_{\Delta}.$
\end{proof}

\section{The measure topology}
\label{sec:tau_meas_top}

In this section, $\mathscr{M}$ denotes a $II_1$ factor acting on the Hilbert space $\mathscr{H}$. Let $\tau$ be the unique faithful normal tracial state on $\mathscr{M}$. For $\varepsilon, \delta > 0$, we define
\begin{equation}
\label{def:tau_top}
\sO_{\tau}(\varepsilon, \delta) := \{ A \in \mathscr{M}:\textrm{ for some projection } E \textrm{ in } \mathscr{M}, \|AE\| \le \varepsilon,\textrm{ and } \tau(I- E) \le \delta \}.
\end{equation} The translation-invariant topology on $\mathscr{M}$ generated by the fundamental system of neighborhoods $\{ \sO_{\tau}(\varepsilon, \delta)\}$ of $0$ is called the \emph{measure topology}. We denote the completion of $\mathscr{M}$ in the measure topology by $\Wtilde{M}$. In \cite{nelson}, Nelson defined the notion of measure topology to study {\it convergence in measure} in a non-commutative setting. In the words of Nelson, the main idea is to ``break up the underlying space into a piece where things behave well plus a small piece.'' We note that in (\ref{def:tau_top}), the projection $I - E$ corresponds to the `small piece'. 

The connection between $\Wtilde{M}$ and $\afm$ becomes apparent from \cite[Theorem 4]{nelson}. Let $A$ be a positive operator in $\afm$ with spectral decomposition $A = \int_{0}^{\infty} \lambda \; dE_{\lambda}$ where $\{ E_{\lambda} \}$ is the resolution of the identity relative to $A$ (the linear manifold $\bigcup _{n \in \mathbb{N}} E_n(\mathscr{H})$ being a core for $A$). For the reader curious about the relevance of the measure topology to $\afm$, it may be helpful to keep in mind that $\{ AE_{\lambda} \}$ is a Cauchy net in $\mathscr{M}$ in the measure topology (which converges to $A$).

In this section, our main goal is to provide an intrinsic characterization of $\afm$ as an ordered complex topological $*$-algebra and study properties of the measure topology. At the outset, we collect some relevant results from \cite{nelson} without proof.

\begin{lemma}[{see \cite[Theorem 1]{nelson}}]
\label{lem:nbd_fund}
\textsl{
For $\varepsilon_1, \delta_1 > 0$ and $\varepsilon_2, \delta_2 > 0$, we have:
\begin{itemize}
\item[(i)] $\sO_{\tau}(\varepsilon_1, \delta_1)^* \subseteq \sO_{\tau}(\varepsilon_1, 2 \delta_1)$,
\item[(ii)] $\sO_{\tau}(\varepsilon_1, \delta_1) + \sO_{\tau}(\varepsilon_2, \delta_2)\subseteq \sO_{\tau}(\varepsilon_1 + \varepsilon_2, \delta_1 + \delta_2)$,
\item[(iii)] $\sO_{\tau}(\varepsilon_1, \delta_1)\sO_{\tau}(\varepsilon_2, \delta_2)\subseteq 
\sO_{\tau}(\varepsilon_1 \varepsilon_2, \delta_1 + \delta_2)$
\end{itemize}
}
\end{lemma}

\begin{cor}
\label{cor:nbd_fund}
\textsl{
For $\varepsilon, \delta > 0,$ let $A \in \sO_{\tau}(\varepsilon, \delta)$ and $B$ be a contraction in $\mathscr{M}$. Then 
\begin{itemize}
\item[(i)] $BA \in \sO_{\tau}(\varepsilon, \delta)$;
\item[(ii)] $AB \in \sO_{\tau}(\varepsilon, 4\delta).$
\end{itemize}
}
\end{cor}
\begin{proof}
\noindent (i) For any projection $E$ in $\mathscr{M}$ such that $\|AE\| \le \varepsilon$, we have $\|(B A) E \| \le \|B\| \|AE\| \le \varepsilon$. Thus if $A \in \sO_{\tau}(\varepsilon, \delta)$, then $BA \in \sO_{\tau}(\varepsilon, \delta)$. 
\vskip 0.1in
\noindent (ii) From Lemma \ref{lem:nbd_fund},(i), we note that $A^* \in \sO_{\tau}(\varepsilon, 2\delta)$. Since $B^*$ is a contraction, by part (i), we observe that $B^*A^* \in \sO_{\tau}(\varepsilon, 2\delta)$. By virtue of Lemma \ref{lem:nbd_fund},(i), we conclude that $AB \in \sO_{\tau}(\varepsilon, 4 \delta)$.
\end{proof}

\begin{thm}[{see \cite[Theorem 1]{nelson}}]
\label{thm:cont_mvn}
\textsl{
The mappings
\begin{align}
A &\mapsto A^* \textrm{ of } \mathscr{M} \to \mathscr{M},\\
(A, B) &\mapsto A+B \textrm{ of } \mathscr{M} \times \mathscr{M} \to \mathscr{M},\\
(A, B) &\mapsto AB \textrm{ of } \mathscr{M} \times \mathscr{M} \to \mathscr{M},
\end{align}
are Cauchy-continuous in the measure topology. Thus the above mappings have unique continuous extensions to $\Wtilde{M}$ giving it the structure of a topological $*$-algebra.}
\end{thm}

\begin{thm}[{see \cite[Theorem 2]{nelson}}]
\label{thm:approx_mvn}
\textsl{
\begin{itemize}
\item[(i)]  $\mathscr{M}$ is Hausdorff in the measure topology. Thus the natural mapping of $\mathscr{M}$ into its completion $\Wtilde{M}$ is an injection.
\item[(ii)] For $A \in \,\Wtilde{M}$ and $\varepsilon > 0$, there is a projection $E$ in $\mathscr{M}$ such that $AE \in \mathscr{M}$ and $\tau(I-E) \le \varepsilon.$
\end{itemize}
}
\end{thm}

\begin{lemma}
\label{lem:approx_sa_mvn}
\textsl{
For a self-adjoint operator $A$ in $\Wtilde{M}$, there is an increasing sequence of projections $\{ E_n \}$ in $\mathscr{M}$ such that $\lim E_n = I$ and $\lim E_nAE_n = A$ in the measure topology.
}
\end{lemma}
\begin{proof}
Let $F$ and $G$ be projections in $\mathscr{M}$. Using projection lattice identities,
\begin{align*}
F(F \wedge G) &= F\wedge G = G(F \wedge G),\\
 F + G &= (F \vee G) + (F \wedge G),
\end{align*}
we observe that $A(F \vee G) = A(F + G) - A(F \wedge G) = A(F+G)(I - \frac{(F\wedge G)}{2}) = (AF + AG)(I - \frac{(F\wedge G)}{2})$. Thus if $F$ and $G$ are projections in $\mathscr{M}$ such that $AF$ and $AG$ belong to $\mathscr{M}$, then $A(F \vee G) \in \mathscr{M}$.

Let $A \in \; \Wtilde{M}^{\!\!\!\textrm{sa}}.$ By Theorem \ref{thm:approx_mvn},(ii), for every $n \in \mathbb{N}$ there is a projection $F_n$ in $\mathscr{M}$ such that $AF_n \in \mathscr{M}$ and $\tau(I-F_n) \le \frac{1}{n}.$ Let $E_n := \vee_{i=1}^n F_i$. We note that $AE_n \in \mathscr{M}$ and $\tau(I-E_n) \le \tau(I-F_n) \le \frac{1}{n}.$ Clearly $E_n \uparrow I$ in the measure topology and $E_nAE_n$ is self-adjoint for all $n \in \mathbb{N}$. By Theorem \ref{thm:cont_mvn}, the sequence  $\{ E_nAE_n \}$ converges to $A$ in the measure topology. Thus $\Wtilde{M}^{\!\textrm{sa}}$ is contained in the measure closure of $\mathscr{M}^{\textrm{sa}}$. 
\end{proof}

\begin{prop}
\label{prop:inc_net_conv}
\textsl{
Let $\{ E_{\alpha} \}$ and $\{ F_{\alpha} \}$ be increasing nets of projections in $\mathscr{M}$ (with the same index set). Let $E := \vee_{\alpha} E_{\alpha}$ and $F := \vee_{\alpha} F_{\alpha}$ and $G$ be a projection in $\mathscr{M}$. Then 
\begin{itemize}
\item[(i)] $\{ E_{\alpha} \vee G \}$ converges in measure to $E \vee G$;
\item[(ii)] $\{ E_{\alpha} \wedge G \}$ converges in measure to $E \wedge G$;
\item[(iii)] $\{ E_{\alpha} \wedge F_{\alpha} \}$ converges in measure to $E \wedge F$.
\end{itemize} 
}
\end{prop}
\begin{proof}
Since $\tau$ is normal, if an increasing net of projections converges to a projection in the strong-operator topology, then the net converges in measure to the same projection. Keeping this in mind, the proof of \cite[Proposition 6.3]{kadison-liu} may be applied here almost verbatim.
\end{proof}

\begin{lemma}
\label{lem:pos_tau_conv}
\textsl{
Let $ \{ A_{\alpha} \}, \{ B_{\alpha} \}$ be nets of positive operators in $\mathscr{M}$ such that $0 \le A_{\alpha} \le B_{\alpha}$. If $\lim B_{\alpha} = 0$ in the measure topology, then $\lim A_{\alpha} = 0$ in the measure topology.
}
\end{lemma}
\begin{proof}
Let $A, B$ be positive operators in $\mathscr{M}$ such that $ 0\le A \le B$ and for $\varepsilon, \delta > 0$, let $B \in \sO_{\tau}(\varepsilon , \delta)$. Take a projection $E$ in $\mathscr{M}$ such that $\|B E \| \le \varepsilon$ and $\tau(I-E) \le \delta.$ As $0 \le EAE \le EBE \le \varepsilon I$, we have $\|EAE\| \le \varepsilon.$ For the projection $F := I - \mathcal{R} \big( A(I-E) \big)$ in $\mathscr{M}$, we have $FA(I-E) = 0$, and using Proposition \ref{prop:range_proj}, (iv), $$I-F = \mathcal{R}(A(I-E)) \sim \mathcal{R}((I-E)A) \Longrightarrow I-F \precsim I-E \Longrightarrow \tau(I-F) \le \tau(I-E).$$ Since $E(E \wedge F) = F(E \wedge F) = E \wedge F$, we have 
$$(I-E)AF = 0 \Longrightarrow AF = EAF \Longrightarrow A(E \wedge F) = EAE(E \wedge F).$$
Thus $\|A(E \wedge F)\| = \|EAE(E \wedge F)\| \le \|EAE\| \le \varepsilon$ and $\tau(I - E \wedge F) = \tau((I-E) \vee (I-F)) \le \tau(I-E) + \tau(I-F) \le 2\tau(I-E) \le 2\delta.$ We conclude that $A \in \sO_{\tau}(\varepsilon, 2\delta).$ Thus in the measure topology, if $\lim B_{\alpha} = 0$, then $ \lim A_{\alpha} = 0$.
\end{proof}

\begin{prop}
\label{prop:mtilde_pos}
\textsl{
\begin{itemize}
\item[(i)] $\mathscr{M}^{\textrm{sa}}$ is a closed subset of $\mathscr{M}$ in the measure topology. The measure closure of $\mathscr{M}^{\textrm{sa}}$ in $\Wtilde{M}$ is $\Wtilde{M}^{\!\textrm{sa}}$, the set of self-adjoint elements in $\Wtilde{M}$.
\item[(ii)] $\mathscr{M}^{+}$ is a closed subset of $\mathscr{M}$ in the measure topology. Let $\Wtilde{M}^{\!\!+}$ denote the measure closure of $\mathscr{M}^{+}$ in $\Wtilde{M}$. Then $\Wtilde{M}^{\!\!+}$ is a cone in $\Wtilde{M}$.
\end{itemize}
The cone $\Wtilde{M}^{\!\!\!+}$ equips $\Wtilde{M} ^{\!\!\mathrm{sa}}$ with a natural order structure making $(\Wtilde{M}; \Wtilde{M}^{\!\!\!\!+})$ an ordered complex topological $*$-algebra.
}
\end{prop}
\begin{proof}
\noindent (i) Let $\{ A_{\alpha} \}_{\alpha \in \Lambda}$ be a net of self-adjoint operators in $\mathscr{M}$ converging to $A \in \,\Wtilde{M}$. By Theorem \ref{thm:cont_mvn}, since the adjoint operation is continuous in the measure topology, we conclude that $A = \lim A_{\alpha} = \lim A_{\alpha}^* = A^*$ which implies that $A$ is self-adjoint. Thus the closure of $\mathscr{M}^{\textrm{sa}}$ is contained in $\Wtilde{M}^{\!\textrm{sa}}$.

Let $A \in \,\Wtilde{M}^{\!\!\textrm{sa}}.$ From Lemma \ref{lem:approx_sa_mvn}, we have an increasing sequence of projections $\{ E_n \}$ such that $\{ E_nAE_n \}$ is a sequence of self-adjoint operators in $\mathscr{M}$ and $\lim E_nAE_n = A$ in the measure topology. Thus $\Wtilde{M}^{\!\!\textrm{sa}}$ is contained in the closure of $\mathscr{M}^{\textrm{sa}}$. 

\noindent (ii)  Let $A \in \mathscr{M}$ and $\{ H_{\alpha} \}$ be a net of positive operators in $\mathscr{M}^{+}$ such that $\lim H_{\alpha} = A$ in the measure topology. By part (i), we note that $A$ is self-adjoint. Let $A_{+}, A_{-}$ denote the positive and negative parts of $A$, respectively, such that $A = A_{+} - A_{-}$. Let $E$ be a spectral projection of $A$ such that $EAE = -A_{-}$. Thus we have $\lim  (EH_{\alpha}E + A_{-}) = 0$ in the measure topology. As $0 \le A_{-} \le EH_{\alpha}E + A_{-}$, by Lemma \ref{lem:pos_tau_conv} we conclude that $A_{-} = 0$ which implies that $A \in \mathscr{M}^{+}$. Hence $\Wtilde{M}^{\!\!+} \cap \mathscr{M} = \mathscr{M}^{+}.$

 For every projection $F$ in $\mathscr{M}$, the mapping $A \in \mathscr{M}^{+} \mapsto FAF \in \mathscr{M}^{+}$ is measure continuous and thus continuously extends to $\Wtilde{M}^{\!\!\!+}$. In other words, for $H \in \,\Wtilde{M}^{\!\!+}$, we have $FHF \in \,\Wtilde{M}^{\!\!+}$.
 
  Let $H \in \,\Wtilde{M}^{\!\!+} \cap (-\Wtilde{M}^{\!\!+})$. From Lemma \ref{lem:approx_sa_mvn}, we have an increasing sequence of projections $\{ E_n \}$ in $\mathscr{M}$ such that $\lim E_n = I$ in the measure topology and $E_nHE_n \in \mathscr{M} \,\cap \Wtilde{M}^{\!\!+} = \mathscr{M}^{+}$ for all $n \in \mathbb{N}$. Similarly as $-H \in \,\Wtilde{M}^{\!\!+}$, we have an increasing sequence of projections $\{ F_n \}$ in $\mathscr{M}$ such that $\lim F_n = I$ in the measure topology and $F_n (-H)F_n \in \mathscr{M}^{+}$ for all $n \in \mathbb{N}$. Let $G_n := E_n \wedge F_n$ for $n \in \mathbb{N}$. Thus $G_nHG_n, -G_nHG_n \in \mathscr{M}^{+}$ which implies that $G_n HG_n = 0$ for all $n \in \mathbb{N}$. By Proposition \ref{prop:inc_net_conv}, (iii), we have $\lim G_n = I$ in the measure topology and hence $H = \lim G_nHG_n = 0$. Consequently, $\Wtilde{M}^{\!\!+} \cap (-\Wtilde{M}^{\!\!+}) = \{0\}.$

 Since by Theorem \ref{thm:cont_mvn} the mapping $+ : \mathscr{M}^{+} \times \mathscr{M}^{+} \to \mathscr{M}^{+} $ is continuous in the measure topology, we note that $\Wtilde{M}^{\!\!\!+}$ is closed under addition. Thus $\Wtilde{M}^{\!\!+}$ is a cone in $\Wtilde{M}$.
\end{proof}

\begin{remark}
From \cite[Theorem 4]{nelson} and the discussion following it in \cite{nelson}, we note that $\Wtilde{M}$ and $\afm$ are isomorphic as unital $*$-algebras with the isomorphism extending the identity mapping of $\mathscr{M} \to \mathscr{M}$. For $A \in \,\Wtilde{M}$, we denote the corresponding operator in $\afm$ by $M_A$. 

\end{remark}

\begin{prop}
\label{prop:pos_equiv}
\textsl{
Let $A$ be a self-adjoint operator in $\Wtilde{M}$. Then $A$ is positive, that is, $A \in \,\Wtilde{M}^{\!\!\!+}$ if and only if $M_A \in \mathscr{M}_{\textrm{aff}}$ is positive.}
\end{prop}

\begin{proof}
Let $A \in \,\Wtilde{M}^{\!\!+}$. From the proof of Proposition \ref{prop:mtilde_pos},(ii) (recall that $\Wtilde{M}^{\!\!+} \cap \mathscr{M} = \mathscr{M}^{+}$), we have an increasing sequence of projections $\{ E_n \}$ in $\mathscr{M}$ which is measure convergent to $I$ such that the sequence of positive operators $\{ E_nAE_n \}$ in $\mathscr{M}$ is measure convergent to $A$. For $n \in \mathbb{N}$, let $F_n$ denote the spectral projection of $M_A \in \afm$ corresponding to the interval $(-n, 0)$, and let $F := \vee_{i=1}^{\infty} F_i$. 

As $M_{E_kAE_k} = \overline{M_{E_k} M_A M_{E_k}} = \overline{E_k M_A E_k}$ is a bounded positive operator in $\mathscr{M}$, for any vector $x \in E_n(\mathscr{H}) \cap \mathscr{D}(M_A)$, we have $\langle M_A x, x \rangle \ge 0.$ Thus $E_n(\mathscr{H}) \cap F_n(\mathscr{H}) = \{0\}$ for all $n \in \mathbb{N}$. Equivalently, $E_n \wedge F_n = 0$ for all $n \in \mathbb{N}$. By Proposition \ref{prop:inc_net_conv}, the sequence $\{ E_n \wedge F_n \}$ is converges in measure to $I \wedge F = F$. Thus $F = 0$.
This shows that $M_A$ is a positive operator in $\afm$. 

We next prove the converse. For $A \in \,\Wtilde{\mathscr{M}}$, let $M_A$ be a positive operator in $\afm$ and for $n \in \N$, let $E_n \in \mathscr{M}$ denote the spectral projection of $M_A$ corresponding to the interval $[0, n]$. Clearly $E_n \uparrow I$ in measure and for each $n \in \N$, $AE_n \in \mathscr{M}$ is a bounded positive operator. Since $\{ AE_n \}$ converges in measure to $A$, we conclude that $A \in \,\Wtilde{M}^{\!\!+}$.
\end{proof}

\begin{cor}
\label{cor:pos_equiv}
\textsl{
Let $A \in \,\Wtilde{M}^{\!\!+}$, and $B \in \Wtilde{M}$. Then $B^*AB$ is in $\Wtilde{M}^{\!\!+}$.
}
\end{cor}
\begin{proof}
Since by Proposition \ref{prop:pos_equiv} $M_A$ is a positive operator in $\afm$, we observe that $M_{B^*AB} = \overline{M_B ^* M_A M_B}$ (see \cite[Theorem 4]{nelson}) is a positive operator in $\afm$. The assertion follows from Proposition \ref{prop:pos_equiv}.
\end{proof}

\begin{remark}
\label{rmrk:maff_pos}
From Proposition \ref{prop:pos_equiv}, the $*$-algebra isomorphism $ A \mapsto M_A :\Wtilde{M} \mapsto\afm$ also induces an order isomorphism of $(\Wtilde{M};\Wtilde{M}^{\!\!+}) $ and $(\afm; \afm ^{+})$. In the rest of the article, we use $\Wtilde{M}$ and $\afm$ interchangeably by transferring the topology of $\Wtilde{M}$ to $\afm$ via the isomorphism. 
\end{remark}

\begin{remark}[Universal property of measure completion]
Let $K$ be a unital ordered {\it complete} complex topological $*$-algebra and consider $\mathscr{M}$ with the measure topology. Let $\iota : \mathscr{M} \to \;\Wtilde{M}$ denote the natural embedding.\ For any continuous order-preserving unital $*$-homomorphism $\Phi : \mathscr{M} \to K$, there is a unique continuous order-preserving unital $*$-homomorphism $\Psi :\; \Wtilde{M} \to K$ such that $\Phi = \Psi \circ \iu.$

Let $\mathscr{M}_1$ and $\mathscr{M}_2$ be $II_1$ factors that are isomorphic as von Neumann algebras, that is, there is a unital $*$-isomorphism $\Phi : \mathscr{M}_1 \to \mathscr{M}_2$ (which is automatically normal). It is straightforward to see from the definition of measure topology that $\Phi$ is measure continuous (considering both $\mathscr{M}_1$ and $\mathscr{M}_2$ with the measure topology.) Thus $\Phi$ induces a continuous unital $*$-isomorphism $\widetilde{\Phi} :\; \widetilde{\mathscr{M}_1} \to  \widetilde{\mathscr{M}_2}$ which is also an order-isomorphism. It is in this sense that we regard $\Wtilde{M}$ as an intrinsic description of $\mathscr{M}_{\textrm{aff}}$.  
\end{remark}

\begin{remark}
\label{rmrk:intrinsic_rank}
Note that a self-adjoint idempotent in $\Wtilde{M}$ is a projection in $\mathscr{M}$ (by the spectral theorem). For a self-adjoint operator $A \in \; \Wtilde{M}$, we define the {\it range projection} of $A$ to be the smallest projection $E$ in $\mathscr{M}$ such that $EA = A$.  If $A$ is not self-adjoint, the range projection of $A$ is defined to be the range projection of $A^*A$. This is compatible with the definition of the range projection for a represented version of $A$ (in $\afm$). Thus Definition \ref{def:rank} gives an intrinsic notion of the rank functional on $\Wtilde{M}$.
\end{remark}

Let $\mathscr{A}$ be a maximal abelian self-adjoint $*$-subalgebra of $\mathscr{M}$ which is generated by a maximal totally ordered set of projections $\{ E_t \}_{t \in [0, 1]}$ in $\mathscr{M}$ such that $\tau(E_s) = s$ for $s \in [0, 1]$.  With $\mu$ denoting the Lebesgue measure on $[0, 1]$, we may view $\mathscr{A}$ as $L^{\infty}([0, 1]; d\mu)$ with the projection $E_t$ corresponding to the characteristic function $\chi_{[0, t]}$. Furthermore, $\tau : \mathscr{A} \to \mathbb{C}$ is given by the mapping $f \in L^{\infty}([0, 1]) \mapsto \int_{0}^{1} f \, d\mu$. The ultraweak topology on $\mathscr{A}$ corresponds to the weak-$*$ topology induced by the predual $L^1([0, 1]; d\mu).$
Recall that for a von Neumann subalgebra $\mathscr{S}$ of $\mathscr{M}$, the ultraweak topology on $\mathscr{S}$ is equivalent to the subspace topology on $\mathscr{S}$ inherited from the ultraweak topology on $\mathscr{M}$. In this framework, we discuss two examples (Example \ref{ex:not_coarser}, \ref{ex:not_finer}) with the goal of answering some natural questions about the measure topology, juxtaposing it with the ultraweak topology. These results are encapsulated in Theorem \ref{thm:top_comp}.

\begin{example}
\label{ex:not_coarser}
For $n \in \mathbb{N}$, define a unitary operator $U_n := \int_{0}^{1} e^{2 \pi \iu n\lambda} \; dE_{\lambda} \in \mathscr{A}$ corresponding to the function $u_n$ in $L^{\infty}([0, 1]; d\mu)$ given by $x \in [0, 1] \mapsto e^{2 \pi \iu n x} \in \mathbb{C}.$ Approximating $f \in L^1([0, 1]; d\mu)$ by step functions and noting that $\int_{0}^{1} e^{2 \pi \iu n x} \chi_{[a, b]}\; d\mu(x) = \frac{1}{2 \pi \iu n}(e^{2 \pi \iu n b} - e^{2 \pi \iu n a}) \longrightarrow 0$ as $n \longrightarrow \infty$ (for $0 \le a < b \le 1$), a standard argument shows that $u_n$ converges in the weak-$*$ topology to $0$. 

 Let $X_n := \bigcup_{k=0}^{n-1} \; [\frac{1}{6n} + \frac{k}{n}, \frac{5}{6n} + \frac{k}{n}] \subset [0, 1]$ and let $F_n$ denote the projection in $\mathscr{A}$ corresponding to the characteristic function of $X_n$. We note that $\tau(F_n) = \mu(X_n) = \frac{2}{3}$ and for $x \in X_n$, 
\begin{equation}
\label{eqn:unitary_conv}
|u_m(x) - u_{m+n}(x)| = |e^{2 \pi \iu nx}  -1 | = 2|\sin ( \pi nx)| \ge 1, \forall m \in \mathbb{N}.
\end{equation}
Let $G$ be a projection in $\mathscr{M}$ such that $\tau(I - G) \le \frac{1}{2}$. As $\tau(G \wedge F_n) = \tau(G) + \tau(F_n) - \tau(G \vee F_n) \ge \frac{1}{2} + \frac{2}{3} - 1 = \frac{1}{6} > 0$, we have $G \wedge F_n \ne 0$. For a unit vector $x \in F_n(\mathscr{H})$, from the inequality in (\ref{eqn:unitary_conv}) we observe that $\|(U_m - U_{m+n})x\| \ge \|x\|$. Thus $\|(U_m - U_{m+n})(G \wedge F_n)\| \ge 1$ and consequently $\|(U_m - U_{m+n})G\| \ge 1.$ We conclude that $U_m - U_{m+n} \notin \sO_{\tau}(\frac{1}{2}, \frac{1}{2})$ for all positive integers $m, n$. In summary,
\begin{itemize}
\item[(i)] $\{ U_n \}$ converges to $0$ in the ultraweak topology;
\item[(ii)] $\{ U_n \}$ is not a Cauchy sequence in the measure topology;
\item[(iii)] $\{ U_n \}$ has no Cauchy subsequences in the measure topology.
\end{itemize}  

\end{example}

\begin{example}
\label{ex:not_finer}
For $n \in \mathbb{N}$, define $H_n := n E_{\frac{1}{n}}$. Since $H_n \in \sO_{\tau}(\frac{1}{n}, \frac{1}{n})$, the sequence $\{ H_n \}$ converges to $0$ in the measure topology. Let $h_n$ be the element of $L^{\infty}([0, 1]; d\mu)$ corresponding to $H_n$. The function $x \in (0, 1] \mapsto \frac{1}{\sqrt{x}} \in [1, \infty)$ is in $L^1([0, 1]; d\mu)$. We note that 
$$\int_{0}^{1} (h_{2n}(x) - h_{n}(x))\frac{1}{\sqrt{x}} \; d\mu(x) = 2(\sqrt{2} - 1) \sqrt{n}.$$
This shows that $\{ H_n \}$ is not a Cauchy sequence in the ultraweak topology.
\end{example}

\begin{thm}
\label{thm:top_comp}
\textsl{ For a $II_1$ factor $\mathscr{M}$, we have the following:
\begin{itemize}
\item[(i)] The measure topology on $\mathscr{M}$ is coarser than the norm topology.
\item[(ii)] The measure topology on $\mathscr{M}$ is neither coarser nor finer than the ultraweak topology on $\mathscr{M}$.
\item[(iii)] The trace $\tau$ is {\bf not} continuous in the measure topology.
\item[(iv)] The restriction of the trace $\tau$ to the unit ball of $\mathscr{M}$ is continuous in the measure topology.
\item[(v)] The unit ball of $\mathscr{M}$ is {\bf not} compact in the measure topology.
\item[(vi)] The unit ball of $\mathscr{M}$ is complete in the measure topology.
\end{itemize}
}
\end{thm}
\begin{proof}
\noindent (i) Let $B(\varepsilon) := \{ A : \|A\| < \varepsilon \}$. For all $\delta > 0$, we observe that $B(\varepsilon) \subseteq \sO_{\tau}(\varepsilon, \delta)$ by using the identity matrix as the projection `$E$' appearing in the definition of $\sO_{\tau}(\varepsilon, \delta)$. Thus the measure topology on $\mathscr{M}$ is coarser than the norm topology. 

In the setting of Example \ref{ex:not_coarser}, since $E_{\frac{1}{n}} \in \sO_{\tau}(\frac{1}{n}, \frac{1}{n} )$, we have $E_{\frac{1}{n}} \rightarrow 0$ in the measure topology as $n \rightarrow \infty$. But clearly $E_{\frac{1}{n}}$ is not a Cauchy sequence in the norm topology. Hence the measure topology, in general, can be different from the norm topology.

\vskip 0.05in

\noindent (ii) From Example \ref{ex:not_coarser}, it follows that the measure topology is not coarser than the ultraweak topology. Similarly, from Example \ref{ex:not_finer}, it follows that the measure topology is not finer than the ultraweak topology.
\vskip 0.05in

\noindent (iii) Consider the sequence of positive operators $\{ H_n \}$ as defined in example \ref{ex:not_finer}. We note that $H_n \longrightarrow 0$ in the measure topology whereas $1 = \tau(H_n) \longrightarrow 1$. Thus $\tau$ is not continuous in the measure topology.
\vskip 0.05in

\noindent (iv) Let $(\mathscr{M})_1$ denote the unit ball of $\mathscr{M}$. For $A \in \sO_{\tau}(\varepsilon, \delta) \cap (\mathscr{M})_1$, let $E$ be a projection in $\mathscr{M}$ such that $\| AE \| \le \varepsilon$ and $\tau(I - E) \le \delta$. Using the Kadison-Schwarz inequality (see \cite[Theorem 1]{kadison-schwarz}), we observe that 
\begin{align*}|\tau(A)| &\le |\tau(AE)| + |\tau(A(I-E))| \\
&\le  \|AE \| + \sqrt{\tau(A^*A)} \cdot \sqrt{\tau(I-E)} \\
&\le \varepsilon + \sqrt{\delta}.\end{align*} Thus the restriction of $\tau$ to $(\mathscr{M})_1$ is continuous in the measure topology.
\vskip 0.05in

\noindent (v) From Example \ref{ex:not_coarser}, the sequence of unitary operators $\{ U_n \}$ has no subsequences that are convergent in the measure topology. In other words, the unit ball of $\mathscr{M}$ is not sequentially compact and hence, not compact.

\vskip 0.05in

\noindent (vi) Let $\{ A_{\alpha} \}$ be a Cauchy net of operators in the unit ball of $\mathscr{M}$ converging to $A \in \; \Wtilde{M}$ in the measure topology. The net of positive operators $\{ A_{\alpha}^* A_{\alpha} \}$ in $\mathscr{M}$ converges to $A^*A$ in the measure topology. Since $\|A_{\alpha}\| \le 1$, we have $A_{\alpha}^*A_{\alpha} \le I$. By Proposition \ref{prop:mtilde_pos},(ii), we note that $A^*A \le I.$ Let $M_A$ be the operator in $\afm$ coresponding to $A$. For a vector $x \in \mathscr{D}(M_{A^*A}) \subseteq \mathscr{D}(M_A)$, we have $\|M_A x\|^2 = \langle M_A x, M_A x \rangle  = \langle M_{A^*A}x, x\rangle \le \|x\|^2$ and thus $M_A$ is a bounded operator in the unit ball of $\mathscr{M}$. Thus every Cauchy net in the unit ball of $\mathscr{M}$ converges to an operator in the unit ball of $\mathscr{M}$.
\end{proof}

\section{Matrix algebras over $\afm$ and $\mathscr{M}_{\Delta}$}

In this section, $\mathscr{M}$ denotes a $II_1$ factor acting on the Hilbert space $\mathscr{H}$ and $\tau$ denotes the unique faithful normal tracial state on $\mathscr{M}$. We note that $M_n(\mathscr{M})$ is a $II_1$ factor with the standard matrix action on $\mathscr{H}^{(n)}$ (:= $\oplus_{i=1}^n \mathscr{H}$). We denote the normalized trace on $M_n(\mathscr{M})$ by $\text{\boldmath$\tau _n$}$. Let $\Delta, \Delta_n$ denote the Fuglede-Kadison determinant on $\mathscr{M}, M_n(\mathscr{M})$, respectively. The $(i, j)$th entry of a matrix $A$ is denoted by $A_{ij}$. The matrix unit which contains $I$ in the $(i, j)$th entry and $0$'s elsewhere is denoted by $\mathcal{E}^{(ij)}.$ In \S \ref{sec:intro}, we briefly discussed the construction of matrix algebras over $*$-algebras. Since $\afm$ and $\mathscr{M}_{\Delta}$ are $*$-algebras, this provides a formal description of $M_n(\afm)$ and $M_n(\mathscr{M}_{\Delta})$. In this section, we provide operator algebraic descriptions of $M_n(\afm)$ and $M_n(\mathscr{M}_{\Delta})$ by showing that we have the isomorphisms $M_n(\afm) \cong M_n(\mathscr{M})_{\textrm{aff}}$, and $M_n(\mathscr{M}_{\Delta}) \cong M_n(\mathscr{M})_{\Delta_n}$, in a natural way.

Let $\mathcal{P}$ denote the product topology on $M_n(\mathscr{M})$ (viewed as $\mathscr{M} \times \overset{n^2}{\cdots} \times \mathscr{M}$) derived from the measure topology on $\mathscr{M}$, and $\mathcal{T}$ denote the measure topology on $M_n(\mathscr{M})$. We note that the topologies $\mathcal{P}$ and $\mathcal{T}$ are both translation-invariant. From our discussion in \S \ref{sec:tau_meas_top}, it is straightforward to see that the completion of $M_n(\mathscr{M})$ in $\mathcal{T}$ may be identified with $M_n(\mathscr{M})_{\textrm{aff}}$ as an ordered unital complex $*$-algebra. In Lemma \ref{lem:product_top}, we observe that the completion of $M_n(\mathscr{M})$ in $\mathcal{P}$ may be identified with $M_n(\afm)$ as a unital complex $*$-algebra.

\begin{lemma}
\label{lem:product_top}
The mappings
\begin{align}
A &\mapsto A^{\dagger} \textrm{ of } M_n(\mathscr{M}) \to M_n(\mathscr{M}),\\
(A, B) &\mapsto A+B \textrm{ of } M_n(\mathscr{M}) \times M_n(\mathscr{M}) \to M_n(\mathscr{M}),\\
(A, B) &\mapsto AB \textrm{ of } M_n(\mathscr{M}) \times M_n(\mathscr{M}) \to M_n(\mathscr{M}),
\end{align}
are continuous in the topology $\mathcal{P}$.\ Thus we may identify the complex $*$-algebra $M_n(\Wtilde{M})$ with the completion of $M_n(\mathscr{M})_{\mathcal{P}}$.
\end{lemma}
\begin{proof}
Follows from the definition of product topology.
\end{proof}

\begin{lemma}
\label{lem:norm_c_star}
\textsl{
Let $\mathfrak{A}$ be a $C^*$-algebra acting on the Hilbert space $\mathscr{H}$ and $n$ be a positive integer. For a matrix $A \in M_n(\mathfrak{A})$, we have:
\begin{itemize}
\item[(i)] $\|A\| \le \sum_{1 \le i, j \le n} \|A_{ij}\|$;
\item[(ii)] $\sup_{1 \le i, j \le n} \|A_{ij}\| \le \|A\|.$
\end{itemize} 
}
\end{lemma}
\begin{proof}
\noindent (i) Consider the $C^*$-algebra $M_n(\mathfrak{A})$ acting on $\mathscr{H}^{(n)}$ by the matrix action on column vectors. Let ${\bf x} = (x_1, \cdots, x_n), {\bf y} = (y_1, \cdots, y_n)$ be unit vectors in $\mathscr{H}^{(n)}$. We note that for $1 \le i \le n$, we have $\|x_i\| \le \|{\bf x}\| = 1, \|y_i\| \le \| {\bf y} \| = 1$. Using the Cauchy-Schwarz inequality, we observe that $$|\langle A {\bf x}, {\bf y} \rangle| = |\sum_{1 \le i, j \le n} \langle A_{ij} x_i, y_j \rangle |  \le \sum_{1 \le i, j \le n} |\langle A_{ij} x_i, y_j \rangle | \le \sum_{1 \le i, j \le n} \| A_{ij}\| \| x_i\| \|y_j \| \le \sum_{1 \le i, j \le n} \| A_{ij} \|.$$
Thus $\|A\| \le \sum_{1 \le i, j \le n} \|A_{ij}\|.$
\vskip 0.05in

\noindent (ii) Fix positive integers $k, \ell \le n$. Let $x, y$ be unit vectors in $\mathscr{H}$. Consider the unit vector ${\bf x} \in \mathscr{H}^{(n)}$ which has $x$ at the $k^{\textrm{th}}$ position and $0$'s elsewhere, and the unit vector ${\bf y} \in \mathscr{H}^{(n)}$ which has $y$ at the $\ell ^{\textrm{th}}$ position and $0$'s elsewhere. Since $\langle A_{k \ell} x, y \rangle = \langle A {\bf x}, {\bf y} \rangle \le \|A\|$, the assertion follows.
\end{proof}

\begin{prop}
\label{prop:equiv_top}
\textsl{
The topologies $\mathcal{P}$ and $\mathcal{T}$ on $M_n(\mathscr{M})$ are equivalent.
}
\end{prop}
\begin{proof}
The subsets of $M_n(\mathscr{M})$ of the form $\prod_{1 \le i, j \le n} \sO_{\tau}(\varepsilon_{ij}, \delta_{ij})$ ($\varepsilon_{ij} > 0, \delta_{ij} > 0$) constitute a fundamental system of neighborhoods of $0$ for the translation-invariant topology $\mathcal{P}$. The subsets of $M_n(\mathscr{M})$ of the form $\sO_{\text{\boldmath$\tau _n$}}(\varepsilon, \delta)$ ($\varepsilon > 0, \delta >0$) constitute a fundamental system of neighborhoods of $0$ for the translation-invariant topology $\mathcal{T}$.

\begin{claim}
\label{clm:P_less_T}
\textsl{
For $\varepsilon_{ij}, \delta_{ij} > 0 \; (1 \le i, j \le n)$, we have
\begin{equation}
\prod_{1 \le i, j \le n} \sO_{\tau}(\varepsilon_{ij}, \delta_{ij}) \subseteq \sO_{\text{\boldmath$\tau _n$}}\Big(\sum_{1 \le i, j \le n} \varepsilon_{ij}, \sum_{1 \le i, j \le n} \delta_{ij} \Big).
\end{equation} }
\end{claim}
\begin{claimpff}
Let $\mathbf{A} \in \prod_{1 \le i, j \le n} \sO_{\tau}(\varepsilon_{ij}, \delta_{ij})$. In other words, $\mathbf{A} \in M_n(\mathscr{M})$ such that $A_{ij} \in  \sO_{\tau}(\varepsilon_{ij}, \delta_{ij})$.
For $1 \le i, j \le n$, there is a projection $E^{(ij)}$ in $\mathscr{M}$ such that $\|A_{ij}E^{(ij)}\| \le \varepsilon_{ij}$ and $\tau(I-E^{(ij)}) \le \delta_{ij}$. Let $E := \wedge_{1 \le i, j \le n} E^{(ij)} \in \mathscr{M}$ and $\mathbf{E}_n := \mathrm{diag}(E, \overset{n}{\cdots}, E) \in M_n(\mathscr{M})$. We note that the $(i, j)$th entry of $\mathbf{A} \mathbf{E}_n$ is given by $A_{ij}E$. Since by Lemma \ref{lem:norm_c_star},(i), $$\|\mathbf{A} \mathbf{E}_n\| \le \sum_{1 \le i, j \le n} \|A_{ij}E\| \le \sum_{1 \le i, j \le n} \|A_{ij}E^{(ij)} \| \le \sum_{1 \le i, j \le n} \varepsilon_{ij},$$ and $$\text{\boldmath$\tau _n$}(\textbf{I}_n - \mathbf{E}_n) = \tau(I - E) = \tau(\vee_{1\le i,j \le n} (I- E^{(ij)})) \le \sum_{1\le i,j \le n} \tau(I - E^{(ij)}) \le \sum_{1 \le i, j \le n} \delta_{ij},$$ we observe that $$A \in \sO_{\text{\boldmath$\tau _n$}}(\sum_{1 \le i, j \le n} \varepsilon_{ij}, \sum_{1 \le i, j \le n} \delta_{ij}).$$
\end{claimpff}

\begin{claim}
\label{clm:T_less_P}
\textsl{
For $0 < \varepsilon$ and $0 < \delta < \frac{1}{16n} $, we have 
\begin{equation}
\sO_{\text{\boldmath$\tau _n$}}(\varepsilon, \delta) \subseteq \prod_{i=1}^{n^2} \sO_{\tau}\big(2\varepsilon, \sqrt{4n \delta}\big).
\end{equation} }
\end{claim}
\begin{claimpff}
Let $\mathbf{A} \in \sO_{\text{\boldmath $\tau _n$}}(\varepsilon, \delta)$ and let $i, j$ be fixed positive integers less than $n$. Since the matrix unit $\mathcal{E}^{(ij)}$ is a contraction in $M_n(\mathscr{M})$, by Corollary \ref{cor:nbd_fund} the matrix $A^{(ij)} := \mathcal{E}^{(1i)}\mathbf{A} \mathcal{E}^{(j1)}$ belongs to $\sO_{\text{\boldmath$\tau _n$}}(\varepsilon, 4\delta)$. Let $\mathbf{E}$ be a projection in $M_n(\mathscr{M})$ such that $\|A^{(ij)} \mathbf{E} \| \le \varepsilon$ and 
\begin{equation}
\label{eqn:contract}
\text{\boldmath$\tau _n$}(\textbf{I}_n - \mathbf{E})  = \frac{1}{n} \left(\tau(I - E_{11}) + \cdots + \tau(I -  E_{nn})\right) \le 4\delta.
\end{equation}
As $A_{ij}E_{11}$ is the $(1, 1)$th entry of $A^{(ij)}\mathbf{E}$, from Lemma \ref{lem:norm_c_star}, (ii),  we have $\|A_{ij}E_{11}\| \le \|A^{(ij)}\mathbf{E} \| \le \varepsilon.$

Since $\mathbf{E}$ is a projection in $M_n(\mathscr{M})$, its $(1, 1)$th entry $E_{11} \in \mathscr{M}$ is a positive contraction. Using the inequality in (\ref{eqn:contract}), we have $1- 4n\delta \le \tau(E_{11}) \le 1$. If the trace of a positive contraction is close to $1$, then the spectrum (with multiplicity given by the trace) is concentrated near $1$. We illustrate this intuitive fact using the estimate below. Let $F_{\lambda}$ denote the resolution of the identity in $\mathscr{M}$ corresponding to the positive operator $E_{11}$. Keeping in mind that$\frac{1}{1-\sqrt{4n\delta}} \le 2$ for $\delta < \frac{1}{16n}$, we have $$\|A_{ij}(I - F_{1-\sqrt{4n\delta}})\| = \|A_{ij}E_{11} \big(\int_{1-\sqrt{4n\delta} }^{1} \frac{1}{\lambda}\;dF_{\lambda}\big)\| \le \frac{1}{1-\sqrt{4n\delta}} \varepsilon \le 2 \varepsilon.$$ 

Let $t := \tau(F_{1-\sqrt{4n \delta}})$. Since $E_{11} \le (1 - \sqrt{4n \delta} ) F_{1 - \sqrt{4n \delta}} + 1 \cdot (I - F_{1- \sqrt{4n \delta}})$, we have $$1 - 4n \delta \le \tau(E_{11}) \le (1-\sqrt{4n \delta})t + (1-t) = 1- \sqrt{4n\delta}\, t  \Longrightarrow \tau(F_{1-\sqrt{4n \delta}}) = t \le \sqrt{4n \delta}.$$ Hence $A_{ij} \in \sO_{\tau}\big(2\varepsilon, \sqrt{4n \delta}\big).$
\end{claimpff}

By Claim \ref{clm:P_less_T} the topology $\mathcal{P}$ is finer than $\mathcal{T}$, and by Claim \ref{clm:T_less_P} the topology $\mathcal{T}$ is finer than $\mathcal{P}$. In conclusion, the topologies $\mathcal{P}$ and $\mathcal{T}$ on $M_n(\mathscr{M})$ are equivalent.
\end{proof}

 From Proposition \ref{prop:mtilde_pos}, the closure of $M_n(\mathscr{M})^{+}$ in the topology $\mathcal{T}$ (and hence in the topology $\mathcal{P}$) is a cone and by Remark \ref{rmrk:maff_pos}, this cone corresponds to the positive cone of $\big( M_n(\mathscr{M})_{\textrm{aff}} \big)^{\mathrm{sa}}$. We encapsulate these observations in the following theorem.

\begin{thm}
\label{thm:main_iso}
\textsl{
 For a positive integer $n$, $M_n(\afm)$ and $M_n(\mathscr{M})_{\textrm{aff}}$ are isomorphic as unital ordered complex  topological $*$-algebras with the isomorphism extending the identity mapping of $M_n(\mathscr{M}) \to M_n(\mathscr{M}).$
}
\end{thm}

\begin{lemma}[Parellelogram inequality]
\label{lem:paralellogram}
\textsl{
Let $P$ be a positive operator in $\afm$. For $n \in \mathbb{N}$ and operators $T_1, \cdots, T_n \in \,\afm$ (and $T := T_1 \afmsum \cdots \afmsum T_n$), we have
$$ T^* \afmprod  P  \afmprod T \le 2^{n-1} \left( \sum_{i=1}^n T_i ^* \afmprod P \afmprod T_i \right).$$
(Note that the summation symbol, $\sum$, is used with respect to $\afmsum$.)
}
\end{lemma}
\begin{proof}

We proceed inductively. For $n=1$, the inequality trivially holds and is in fact an identity. For $n = 2$, we note that 
\begin{align*}
0 &\le (T_1 \afmdiff T_2)^* \afmprod P \afmprod (T_1 \afmdiff T_2)\\
&  = T_1^* \afmprod P \afmprod T_1 \afmsum T_2 ^* \afmprod P \afmprod T_2 \afmdiff (T_1^* \afmprod P \afmprod T_2 \afmsum T_2^* \afmprod P \afmprod T_1)\\
\Longrightarrow  T_1^* \afmprod P \afmprod T_2 \afmsum T_2^* \afmprod P \afmprod T_1 &\le T_1^* \afmprod P \afmprod T_1 \afmsum T_2 ^* \afmprod P \afmprod T_2\\
\Longrightarrow  (T_1 \afmsum T_2)^* \afmprod  P \afmprod (T_1 \afmsum T_2) &\le 2(T_1^* \afmprod P \afmprod T_1 \afmsum T_2 ^* \afmprod P \afmprod T_2).
\end{align*}
Thus writing $T$ as $(T \afmdiff T_n) \afmsum T_n$, we observe that 
\begin{equation}
\label{eqn:parallelogram_1}
T^* \afmprod P \afmprod T \le 2 \left( (T \afmdiff T_n)^* \afmprod P \afmprod (T \afmdiff T_n) \afmsum T_n^* \afmprod P \afmprod T_n \right).
\end{equation}
By the induction hypothesis, we have 
\begin{equation}
\label{eqn:parallelogram_2}
(T \afmdiff T_n)^* \afmprod P \afmprod (T \afmdiff T_n) \le 2^{n-2} \left( \sum_{i=1}^{n-1} T_i ^* \afmprod P \afmprod  T_i \right).
\end{equation}
Combining inequalities (\ref{eqn:parallelogram_1}), (\ref{eqn:parallelogram_2}), and the fact that $T_n^* \afmprod P \afmprod T_n \le 2^{n-2}(T_n^* \afmprod P \afmprod T_n)$ for $n \ge 2$, we reach the desired conclusion.
\end{proof}

\begin{thm}
\label{thm:main_iso_delta}
\textsl{
For $n \in \mathbb{N}$, we have $M_n(\mathscr{M})_{\Delta_n} = M_n(\mathscr{M}_{\Delta})$.
}
\end{thm}
\begin{proof}
Recall that $\mathscr{M}_{\Delta} \subset \afm$ and by Theorem \ref{thm:main_iso}, we may identify $M_n(\afm)$ and $M_n(\mathscr{M})_{\textrm{aff}}$ with each other. Keeping this in mind, we consider the $*$-algebras $M_n(\mathscr{M})_{\Delta_n}$ and $M_n(\mathscr{M}_{\Delta})$ as $*$-subalgebras of $M_n(\mathscr{M})_{\textrm{aff}}$. 
\vskip 0.1in
\begin{claim}
$M_n(\mathscr{M})_{\Delta _n} \subseteq M_n(\mathscr{M}_{\Delta}).$
\end{claim}
\begin{claimpff}
Let $A \in M_n(\mathscr{M})_{\Delta _n}.$  Since the matrix units are bounded operators, by Lemma \ref{lem:mdelta_desc}, note that $\mathrm{diag}(A_{ij}, 0, \cdots, 0) = \mathcal{E}^{(1i)} A\, \mathcal{E}^{(j1)}$ belongs to $ M_n(\mathscr{M})_{\Delta _n}$. Thus $$\tau(\log ^{+} (|A_{ij}|)) = n \tau_n(\log ^{+} |\mathrm{diag}(A_{ij}, 0, \cdots, 0)|) < \infty,$$ from which we conclude that $A_{ij} \in \mathscr{M}_{\Delta}$ for $1 \le i, j \le n$ and $A \in M_n(\mathscr{M}_{\Delta}).$
\end{claimpff}

\begin{claim}
$M_n(\mathscr{M}_{\Delta}) \subseteq M_n(\mathscr{M})_{\Delta _n}.$
\end{claim}
\begin{claimpff}
Let $A \in M_n(\mathscr{M}_{\Delta}) \subset M_n(\afm) \cong M_n(\mathscr{M})_{\textrm{aff}}.$ By the polar decomposition theorem, we have a unitary operator $U \in M_n(\mathscr{M})$ and a positive operator $P$ in $M_n(\mathscr{M})_{\textrm{aff}}$ such that $A = U P$. As $P = U^* A$, from Lemma \ref{lem:mdelta_desc} we observe that $P_{ij} = \sum_{k=1}^n U_{ki}^*A_{kj}$ belongs to $\mathscr{M}_{\Delta}$ for $1 \le i, j \le n $. Thus $P \in M_n(\mathscr{M}_{\Delta})$. Since $$\tau_n(\log ^{+} \mathrm{diag}(P_{11}, \cdots, P_{nn})) = \frac{1}{n}(\tau(\log ^{+} P_{11}) + \cdots + \tau(\log ^{+} P_{nn})) < \infty,$$ we have $2^{n-1} \mathrm{diag}(P_{11}, \cdots, P_{nn}) \in M_n(\mathscr{M})_{\Delta _n}$. Since $\sum_{i=1}^n (\mathcal{E}^{(ii)})^2 = \sum_{i=1}^n \mathcal{E}^{(ii)} = I$, by virtue of the parallelogram inequality (see Lemma \ref{lem:paralellogram}) we have $$P \le 2^{n-1}\left(\sum_{i=1}^n \mathcal{E}^{(ii)}P\mathcal{E}^{(ii)} \right) = 2^{n-1} \mathrm{diag}(P_{11}, \cdots, P_{nn}).$$ (Note that the summation symbol $\sum$ is used above with respect to $\afmsum$.)  From Proposition \ref{prop:ord_ideal_mdelta}, we conclude that $P \in M_n(\mathscr{M})_{\Delta _n}$ and thus $A = UP \in M_n(\mathscr{M})_{\Delta _n}$.
\end{claimpff}
\end{proof}

\section{Applications to the Heisenberg relation}

\label{sec:heisenberg}

Let $\mathscr{M}$ be a $II_1$ factor acting on the Hilbert space $\mathscr{H}$. In this section, we apply the results proved in earlier sections to provide some necessary conditions for pairs of operators $P, Q$ in $\afm$ satisfying the Heisenberg commutation relation. From Remark \ref{rmrk:intrinsic_rank}, recall that the rank functional on $\afm$ is independent of the representation of $\mathscr{M}$.

\begin{lemma}[{cf.\ \cite[Proposition 2.24]{haagerup-schultz}}]
\label{lem:upp_tri}
\textsl{
For $A \in \mathscr{M}_{\Delta}$, we have 
$$
\Delta_2 \Big(
\begin{bmatrix}
I & A\\
0 & I
\end{bmatrix} 
\Big) = 
\Delta_2 \Big(
\begin{bmatrix}
I & 0\\
A & I
\end{bmatrix} 
\Big) = 1.$$
}
\end{lemma}
\begin{proof}
The result follows from the algebraic identity,
$$
\begin{bmatrix}
I & A\\
0 & I
\end{bmatrix} ^{-1} = 
\begin{bmatrix}
I & -A\\
0 & I
\end{bmatrix} = 
\begin{bmatrix}
I & 0\\
0 & -I
\end{bmatrix}\begin{bmatrix}
I & A\\
0 & I
\end{bmatrix}\begin{bmatrix}
I & 0\\
0 & -I
\end{bmatrix}.
$$
Note that here we have used Theorem \ref{thm:main_iso_delta} to conclude that $$\begin{bmatrix}
I & A\\
0 & I
\end{bmatrix},
\begin{bmatrix}
I & -A\\
0 & I
\end{bmatrix} \in M_2(\mathscr{M})_{\Delta_2}.$$
\end{proof}

The following algebraic identity involving free indeterminates $x, y$, is key to our results concerning the Heisenberg-von Neumann puzzle.
\begin{equation}
\label{eqn:mat_id}
\ \left[ \begin{array}{ccc}
1 &  x \\
0 & 1 \end{array} \right]\ \left[ \begin{array}{ccc}
1-xy & 0 \\
0 & 1 \end{array} \right]\ \left[ \begin{array}{ccc}
1 & 0 \\
y & 1 \end{array} \right]
 =\ \left[ \begin{array}{ccc}
1 &  0 \\
y & 1 \end{array} \right]\ \left[ \begin{array}{ccc}
1 & 0 \\
0 & 1-yx \end{array} \right]\ \left[ \begin{array}{ccc}
1 & x \\
0 & 1 \end{array} \right].
\end{equation}

Recall that for $T \in \mathscr{M}_{\Delta}$ we denote the Brown measure of $T$ by $\mu _T$. 

\begin{thm}
\label{thm:bmeas_ab_ba}
\textsl{
\begin{itemize}
\item[(i)] For operators $A, B$ in $\afm$, we have $$\fr(zI \afmdiff A \afmprod B) = \fr(zI \afmdiff B \afmprod A), \;\forall z \in \mathbb{C} - \{0\}.$$
\item[(ii)] For operators $A, B$ in $\mathscr{M}_{\Delta}$, we have $$\Delta(zI \afmdiff A \afmprod B) = \Delta(zI \afmdiff  B \afmprod A), \;\forall z \in \mathbb{C},$$
and thus $$\mu_{A \afmprod B} = \mu_{B \afmprod A}.$$
\end{itemize}
}
\end{thm}
\vskip 0.05in
\begin{proof}
\noindent (i) Since $A, B \in \afm$, from Theorem \ref{thm:main_iso} we observe that the operators,  $$\begin{bmatrix}
I & A\\
0 & I
\end{bmatrix},
\begin{bmatrix}
I & 0\\
B & I
\end{bmatrix} \in M_2(\mathscr{M})_{\textrm{aff}},$$ and are invertible in $M_2(\mathscr{M})_{\textrm{aff}}$ with inverses $$\begin{bmatrix}
I & -A\\
0 & I
\end{bmatrix},
\begin{bmatrix}
I & 0\\
-B & I
\end{bmatrix}, \textrm{respectively}.$$
 Thus evaluating the rank functional for $M_2(\mathscr{M})_{\textrm{aff}}$ on both sides of the identity in (\ref{eqn:mat_id}) (substituting $x = A, y = B$), from Proposition \ref{prop:rank_prop} we conclude that $\fr(I \afmdiff A \afmprod B) + \fr (I) = \fr(I \afmdiff B \afmprod A) + \fr(I)$ which implies that $\fr (I \afmdiff A \afmprod B) = \fr (I \afmdiff B \afmprod A)$. For $z \ne 0$, $\fr (zI \afmdiff A \afmprod B) = \fr (I \afmdiff (z^{-1}A) \afmprod B)= \fr (I \afmdiff  B \afmprod (z^{-1}A)) = \fr (zI \afmdiff B \afmprod A).$

\vskip 0.1in
\noindent (ii) Since $A, B \in \mathscr{M}_{\Delta}$, from Theorem \ref{thm:main_iso_delta} we observe that  $$\begin{bmatrix}
I & A\\
0 & I
\end{bmatrix},
\begin{bmatrix}
I & 0\\
B & I
\end{bmatrix} \in M_2(\mathscr{M})_{\Delta _2}.$$

Using Lemma \ref{lem:upp_tri} to evaluate $\Delta_2$ on both sides of the identity in (\ref{eqn:mat_id}) (substituting $x = A, y = B$), we note that $\Delta(I \afmdiff A \afmprod B) = \Delta(I \afmdiff B \afmprod A)$ for $A, B \in \mathscr{M}_{\Delta}$. For $z \ne 0$,  $\Delta(zI \afmdiff A \afmprod B) = |z| \Delta(I \afmdiff (\frac{1}{z}A) \afmprod B) = |z|\Delta(I \afmdiff  B \afmprod (\frac{1}{z}A)) = \Delta(zI \afmdiff B \afmprod A).$ For $z = 0$, using \cite[Proposition 2.5]{haagerup-schultz}, we have $\Delta(-A \afmprod B) = \Delta(A)\Delta(B) = \Delta(B)\Delta(A) = \Delta(-B \afmprod A).$ Taking the Laplacian of the mappings $z \in \mathbb{C} \mapsto \frac{1}{2\pi} \log \Delta(zI \afmdiff A \afmprod B)$, and $z \in \mathbb{C} \mapsto \frac{1}{2\pi} \log \Delta(zI \afmdiff B \afmprod A)$, we conclude that $\mu_{A \afmprod B} = \mu_{B \afmprod A}.$
\end{proof}

\begin{cor}
\label{cor:heisenberg_lp}
\textsl{
Let $P, Q$ be operators in $\afm$ such that $Q \afmprod P \afmdiff P \afmprod Q = \iu I$. Then at least one of $P$ or $Q$ does not belong to $\mathscr{M}_{\Delta}$ and {\it a fortiori}, does not belong to $L^p(\mathscr{M}, \tau)$  for any $p \in (0, \infty]$.
}
\end{cor}
\begin{proof}
Let, if possible, $P$ and $Q$ be operators in $\mathscr{M}_{\Delta}$ such that $Q \afmprod P \afmdiff P \afmprod Q = \iu I$. Absorbing $-\iu$ into one of the operators, we may assume that $Q \afmprod P \afmdiff P \afmprod Q = I.$ For $w \in \mathbb{C}$, denote the open disc of radius $\frac{1}{2}$ in $\mathbb{C}$ centered at $w$ by $B_w := \{ z : |z - w| < \frac{1}{2} \} \subset \mathbb{C}$. From Theorem \ref{thm:bmeas_ab_ba}, we observe that $\mu_{Q \afmprod P}(B_w) = \mu_{P \afmprod Q}(B_w) =  \mu_{Q \afmprod P \afmdiff I}(B_w) = \mu_{Q \afmprod P}(B_{w + 1})$ for all $w \in \mathbb{C}$. Since $\mu_{Q \afmprod P}$ is a Borel probability measure, there is $\lambda \in \mathbb{C}$ such that $\mu_{Q \afmprod P}(B_{\lambda}) > 0.$ As $\{ B_{\lambda + n-1} \}_{n \in \mathbb{N}}$ is a collection of mutually disjoint open unit discs, note that $\mu_{Q \afmprod P}(\bigcup_{k=0}^{n-1} B_{\lambda + k}) = n \,\mu_{Q \afmprod P}(B_{\lambda})$ for $n \in \mathbb{N}$. Thus by choosing $n$ to be sufficiently large, we have $\mu_{Q \afmprod P}(\bigcup_{k=0}^{n-1} B_{\lambda + k}) > 1$, contradicting the fact that $\mu_{Q \afmprod P}$ is a probability measure. Thus at least one of $P$ or $Q$ does not belong to $\mathscr{M}_{\Delta} \supset \bigcup_{p \in (0, \infty]} \; L^p(\mathscr{M}, \tau)$.
\end{proof}

\begin{cor}
\label{cor:heisenberg_drange}
\textsl{
Let $P, Q$ be operators in $\afm$ such that $Q \afmprod P \afmdiff P \afmprod Q = \iu I$. Then for all $\lambda \in \C$, the operators $P - \lambda I$ and $Q - \lambda I$ are invertible in $\afm$, that is, the respective point spectrums of $P$ and $Q$ are empty.
}
\end{cor}
\begin{proof}
Absorbing $-\iu$ into one of the operators, we may assume that $Q \afmprod P \afmdiff P \afmprod Q=I.$
By Theorem \ref{thm:bmeas_ab_ba},(i), we have $$\fr(nI \afmdiff P^* \afmprod Q^*) = \fr(nI \afmdiff Q^* \afmprod P^*), \textrm{ for }n \in \Z - \{ 0 \}.$$ For an operator $T \in \afm$, define $\fn(T) := \tau(\mathcal{N}(T)) $. Using Proposition \ref{prop:range_proj},(i), we also have $$\fn(nI \afmdiff Q \afmprod P) = \fn(nI \afmdiff P \afmprod Q), \textrm{ for }n \in \Z - \{ 0 \}.$$
As $Q \afmprod P \afmdiff P \afmprod Q = I$, for $k \in \N$ we observe that $\fn(kI \afmdiff  Q \afmprod P) = \fn((k-1)I \afmdiff P \afmprod Q) = \fn((k-1)I \afmdiff Q \afmprod P),$ and $\fn(-kI \afmdiff P \afmprod Q) = \fn((-k + 1)I \afmdiff Q \afmprod P) = \fn((-k + 1)I \afmdiff P \afmprod Q)$. By induction we conclude that $\fn(kI \afmdiff Q \afmprod P) = \fn(Q \afmprod P)$, and $\fn(-k I \afmdiff P \afmprod Q) = \fn(P \afmprod Q)$.
 
 For $k \in \mathbb{N}$, define $E_k := \mathcal{N}(k I \afmdiff Q \afmprod P)$. Note that for distinct positive integers $k$ and $\ell$, we have $E_k \wedge E_{\ell} = 0$. For a projection $E \in \mathscr{M}$, let $\mathrm{dim}_c (E) := \tau(E)$ denote the normalized dimension of $E$. Recall that $\mathrm{dim}_c(E \vee F) = \mathrm{dim}_c(E) + \mathrm{dim}_c(F)$ for projections $E, F$ such that $E \wedge F = 0$. Thus for $n \in \mathbb{N}$ we have
\begin{align*}
n \; \mathrm{dim}_c(\mathcal{N}(Q \afmprod P)) &= \sum_{k=1}^n \mathrm{dim}_c(E_k)\\
&= \mathrm{dim}_c(\vee_{k=1}^n E_k)  \\
&\le \mathrm{dim}_c(I) = 1.
\end{align*}

Consequently, $$\mathrm{dim}_c(\mathcal{N}(Q \afmprod P)) = 0 \Longrightarrow \mathcal{N}(Q \afmprod P) = 0 \Longrightarrow \mathcal{N}(P) = 0.$$ Using an analogous argument for the projections $\mathcal{N}(-mI \afmdiff P \afmprod Q)$ ($m \in \N$), we observe that $\mathcal{N}(Q) = 0$. Taking adjoints of both sides of the relation $Q \afmprod P \afmdiff P \afmprod Q = I$, we get the relation $P^* \afmprod Q^* \afmdiff  Q^* \afmprod P^* = I$. Thus $\mathcal{N}(P^*) = 0,  \; \mathcal{N}(Q^*) = 0$ which by Proposition \ref{prop:range_proj},(i), implies that $\mathcal{R}(P) = I, \mathcal{R}(Q) = I$. By Proposition \ref{prop:luck}, $P$ and $Q$ are invertible in $\afm$. In other words, we have shown that if the operators $P, Q$ satisfy the relation $Q \afmprod P \afmdiff P \afmprod Q = I$, then $P$ and $Q$ are invertible in $\afm$. The proof is complete after noting that if $Q \afmprod P \afmdiff  P \afmprod Q = I$, then $(Q - \lambda I) \afmprod (P - \lambda I) \afmdiff (P - \lambda I) \afmprod (Q - \lambda I) = I$ for all $\lambda \in \C$.
\end{proof}

By virtue of Corollary \ref{cor:heisenberg_drange}, the Heisenberg-von Neumann puzzle may be recasted in the following equivalent manner.
\begin{qn}
\textsl{Are there invertible operators $P, A$ in $\mathscr{M}_{\textrm{aff}}$ such that $$P^{-1} \afmprod A \afmprod P = I \afmsum A ?$$}
\end{qn}

\bibliographystyle{plain}
\bibliography{matrix_algebras_over_algebras_of_unbounded_operators_Nayak.bib}

\end{document}